\documentclass[a4paper, 11pt]{article}

\newif \ifjournal

\makeatletter
\@ifclassloaded{imsart}
  {\journaltrue
  \usepackage{xr-hyper}}
  {\journalfalse}
\makeatother

\RequirePackage[OT1]{fontenc}
\RequirePackage{amsthm,amsmath}
\RequirePackage[colorlinks,
linkcolor=blue,citecolor=blue,urlcolor=blue]{hyperref}
\usepackage{graphicx}
\usepackage[dvipsnames]{xcolor}
\usepackage{amssymb}
\usepackage{booktabs}
\usepackage{tcolorbox}
\usepackage{tagging}
\usepackage{float}
\usepackage{blkarray}
\graphicspath{{images/}{tikz/}}
\usepackage{tikz}
\usepackage{standalone}
\usepackage{gincltex}
\usepackage{./tikz/style/mytikzstyles}
\usepackage{array}
\usepackage{multirow}

\makeatletter
\@ifclassloaded{imsart}
  {\RequirePackage[authoryear]{natbib}
   \arxiv{arXiv:1908.05097}
   \externaldocument{suppl_aos}
   \urlstyle{same}
  }
  {\usepackage[numbers]{natbib}
   \usepackage{fullpage}
   \usepackage{times}
   \urlstyle{same}
  }
\makeatother

\ifjournal{\startlocaldefs}\fi
\renewcommand{\P}{\mbox{P}}
\newcommand{\E}{\mbox{E}}
\newcommand{\Var}{\mathrm{Var}}
\newcommand{\X}{(X_1,X_2)}

\newcommand{\RV}{\mbox{RV}}

\newcommand{\pa}{\mathrm{pa}}

\newcommand{\an}{\mathrm{an}}
\newcommand{\An}{\mathrm{An}}

\newcommand{\G}{\Gamma}

\newcommand{\eps}{\varepsilon}
\newcommand{\oneto}[1]{\{1, \dots, #1\}}

\newcommand{\matr}[1]{\mathbf{#1}}

\newtheorem{prop}{Proposition}
\newtheorem{cor}{Corollary}
\newtheorem{definition}{Definition}
\newtheorem{lemma}{Lemma}
\newtheorem{theorem}{Theorem}
\theoremstyle{remark}
\newtheorem{example}{Example}
\theoremstyle{definition}
\newtheorem{remark}{Remark}

\newcommand{\myabstract}{
  \begin{abstract}
  Causal questions are omnipresent in many scientific problems. While much progress has been made in the analysis of causal relationships between random variables, these methods are not well suited if the causal mechanisms only manifest themselves in extremes. This work aims to connect the two fields of causal inference and extreme value theory. We define the causal tail coefficient that captures asymmetries in the extremal dependence of two random variables. In the population case, the causal tail coefficient is shown to reveal the causal structure if the distribution follows a linear structural causal model. This holds even in the presence of latent common causes that have the same tail index as the observed variables. Based on a consistent estimator of the causal tail coefficient, we propose a computationally highly efficient algorithm that estimates the causal structure. We prove that our method consistently recovers the causal order and we compare it to other well-established and non-extremal approaches in causal discovery on synthetic and real data. The code is available as an open-access \texttt{R} package.
  \end{abstract}
}
\ifjournal{\endlocaldefs}\fi

\newif \ifmyswitch
\myswitchtrue %

\usetag{}

\ifjournal{}
\else{
  \title{Causal discovery in heavy-tailed models}

  \author{
  Nicola Gnecco\\
  University of Geneva\\
  \url{nicola.gnecco@unige.ch}
  \and
  Nicolai Meinshausen\\
  ETH Zurich\\
  \url{meinshausen@stat.math.ethz.ch}
  \and
  Jonas Peters\\
  University of Copenhagen\\
  \url{jonas.peters@math.ku.dk}
  \and
  Sebastian Engelke\\
  University of Geneva\\
  \url{sebastian.engelke@unige.ch}}
  \date{\today}
}\fi

\begin{document}

\ifjournal{
  \begin{frontmatter}

  \title{Causal discovery in heavy-tailed models}
  \runtitle{Causal discovery in heavy-tailed models}

  \begin{aug}
  \author[A]{\fnms{Nicola} \snm{Gnecco}\ead[label=e4,mark]{nicola.gnecco@unige.ch}},
  \author[B]{\fnms{Nicolai} \snm{Meinshausen}\ead[label=e2]{meinshausen@stat.math.ethz.ch}},\\
  \author[C]{\fnms{Jonas} \snm{Peters}\ead[label=e3]{jonas.peters@math.ku.dk}}
  \and
  \author[A]{\fnms{Sebastian} \snm{Engelke}\ead[label=e1,mark]{sebastian.engelke@unige.ch}}

  \address[A]{Research Center for Statistics,
  University of Geneva,
  \printead{e4};
  \printead{e1}
  }
  \address[B]{Department of Mathematics,
  ETH Zurich,
  \printead{e2}
  }
  \address[C]{Department of Mathematical Sciences,
  University of Copenhagen,
  \printead{e3}
  }
  \end{aug}

  \myabstract

  
  
  \end{frontmatter}
}\else{
  \maketitle
  \myabstract
}\fi

\section{Introduction and background}
Reasoning about the causal structure underlying a data generating
process is a key scientific question in many disciplines.
In recent years, much progress has been made in the formalisation
of causal language \citep{Pearl2009, spirtes2000causation, Imbens2015}.
In several situations, causal
relationships
manifest themselves only in
extreme events.
As stated by~\citet[Sec.\ 8.7]{cox1996multivariate}, large interventions (also named natural
experiments) often carry information that is likely to be causal.
In this light, one can view extreme observations as natural experiments
that perturb a given system and facilitate causal analysis.
On the one hand, existing causal methodology,
focuses on moment related quantities of the distribution
and
is not tailored
to estimate causal relationships from extreme events.
On the other hand,
the statistics of univariate extremes
is relatively well understood
\citep{res2008} and there is a large set of tools for the analysis
of heavy-tailed distributions.
This work attempts to
bring
the fields of causality and extremes closer.

Let us first consider a
bivariate random vector $\X$
and assume that we are interested in the causal relationship between the two random variables, $X_1$ and $X_2$.
We consider a linear structural causal model \citep[Sec.\ 1.4]{Pearl2009}
over variables including $\X$ (without feedback mechanisms).
We can then distinguish between the
six scenarios of causal configurations
shown in Figure~\ref{fig:6cases}
that include $X_1$, $X_2$ and possibly a third unobserved random variable $X_0$.
This collection is complete in the following sense.
Any structural causal model
including $X_1$ and $X_2$
is interventionally equivalent
\citep[e.g.,][Sec.\ 6.8]{Peters2017book}
to one of the examples shown in Figure~\ref{fig:6cases}
when taking into account
interventions on $X_1$ or $X_2$ only.
The dashed edges can be interpreted as directed paths induced by a linear structural causal model (SCM);
see Section~\ref{sec:scms} for a formal definition.
\begin{figure}[!ht]
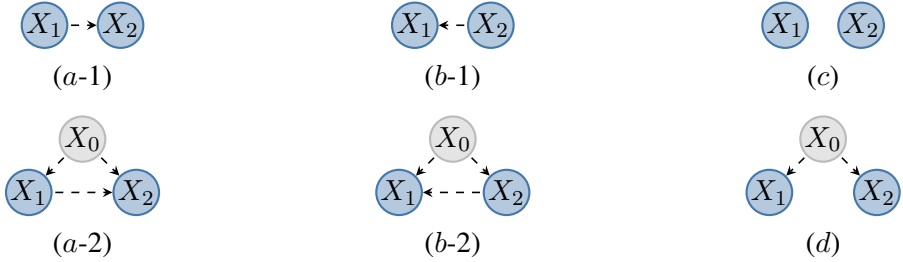

\begin{minipage}{0.3\linewidth}
  \centering
  \includegraphics{x1_x2.tex}\\
  ($a$-1)
\end{minipage}
\begin{minipage}{0.3\linewidth}
  \centering
  \includegraphics{x2_x1.tex}\\
  ($b$-1)
\end{minipage}
\begin{minipage}{0.3\linewidth}
  \centering
  \includegraphics{x1x2.tex}\\
  ($c$)
\end{minipage}

\vspace{3.00mm}

\begin{minipage}{0.3\linewidth}
  \centering
  \includegraphics{x1_xcx2.tex}\\
  ($a$-2)
\end{minipage}
\begin{minipage}{0.3\linewidth}
  \centering
  \includegraphics{x1xc_x2.tex}\\
  ($b$-2)
\end{minipage}
\begin{minipage}{0.3\linewidth}
  \centering
  \includegraphics{x1xcx2.tex}\\
  ($d$)
\end{minipage}
\caption{\label{fig:6cases}
The six possible causal configurations between $X_1$ and $X_2$, and possibly a third
unobserved variable $X_0$.
The variable $X_0$ will be referred to as a hidden confounder. Formal definitions are included in Section~\ref{sec:scms}.
We will see in Section~\ref{subsec:propGamma} that both configurations ($a$-1) and ($a$-2), for example,
show the same tail coefficient behaviour.
The enumeration letters ($a$)--($d$) visualise the cases in Table~\ref{tab:caus}.
}
\end{figure}
Here and in the sequel, we say that
``$X_1$ is the cause of $X_2$'' or
``$X_1$ causes $X_2$''
if there is a directed path from $X_1$ to $X_2$ in the
SCM's
underlying directed acyclic graph (DAG).
Assume that $X_1$ is the cause of $X_2$, and
that both variables are heavy-tailed.
Intuitively,
if the causal relationship
is monotonic, then an
extremely large value of $X_1$ should cause an
extreme value of $X_2$. The causal direction should, therefore,
be strongly visible in the largest absolute values of the random variables.
Here, ``extreme'' is to be seen in the respective scale of each
variable, so it will make sense to consider the copula $\{F_1(X_1), F_2(X_2)\}$, where
$F_j$ is the marginal distribution of $X_j$, $j=1,2$.
To exploit this intuition, we
define the \emph{causal tail coefficient} between variables $X_1$ and $X_2$ as
\begin{align}\label{ctc}
  \G_{12} := \lim_{u\to 1^-} \E\left[F_2(X_2)  \mid F_1(X_1) > u \right] \in [0,1],
\end{align}
if the limit exists. It reflects the causal relationship between
$X_1$ and $X_2$ since, intuitively, if $X_1$ has a
monotonically increasing causal influence on $X_2$,
we expect $\G_{12}$ to be close to one. Conversely,
extremes of $X_2$ will not necessarily lead to extremes of $X_1$ and therefore,
the coefficient $\G_{21}$, where the roles of $X_1$ and $X_2$ in \eqref{ctc} are reversed,
may be strictly smaller than one. This asymmetry
will be made precise in Theorem~\ref{thm:6cases} for linear structural causal models with heavy-tailed noise variables, and it forms the basis of our causal discovery
algorithm.

A different perspective of our approach goes beyond the usual way of defining causality in the bulk of distribution. Namely, by looking at the signal in the tails, we might recover an extremal causal mechanism that is not necessarily present in the central part of the distribution.
An example of this can be observed in financial markets.
During calm periods, it is not clear whether any causal relationship exists among the financial variables.
However, during turmoil, it is common to observe one specific stock or sector causing very negative (or positive) returns of other stocks (or sectors), displaying an extremal causal mechanism.
In the finance literature, the concept of extremal causal mechanism is explained in terms of contagion, i.e., the spread of shocks across different markets~\citep[see][]{forbes2002no}. For example, \citet{rodriguez2007measuring} uses a copula model with time-varying parameters to explain such contagion phenomena within countries in Latin America and Asia.
On the other hand, there are also applications where the causal mechanism is present in the bulk of the distribution but absent in the tails.
\citet{seneviratne2010investigating} present this type of causal relationship between air temperature and the evapotranspiration of soil moisture.
As the air temperature increases,
the evapotranspiration process increases, too. This continues until the soil moisture resources are reduced to the point that a further increase in the temperature has no causal effect on the evapotranspiration.

Heavy-tailed distributions are an example of non-Gaussian models, which have received some attention in the causal literature.
The
LiNGAM
algorithm \citep{shimizu2006linear} exploits
non-Gaussianity through independent component analysis \citep{comon1994independent} to estimate an underlying causal structure.
\citet{mis2016} consider stable noise variables in a Bayesian network and develop a
structure learning algorithm based on BIC.
The work of \citet{gissibl2019identifiability}
also
studies causal questions related to extreme events.
They consider max-linear models \citep{gissibl2018max} where only the largest effect
propagates to the descendants in a Bayesian network.
The work by~\citet{naveau2018revising} falls within the domain of attribution science. Namely, by studying extreme climate events, they try to answer
counterfactual questions such as ``what the Earth's climate might have been without anthropogenic interventions''.
\citet{mhalla2019causal} develop a method to estimate the causal relationships between bivariate extreme observations.
It
relies on the Kolmogorov complexity concept \citep[see][]{kolmogorov1965three} adapted to high conditional quantiles. \cite{eng2021} reviews recent work on causality and sparsity in extremes.

The rest of the paper is organised as follows.
Sections~\ref{sec:scms} and~\ref{sec:evt} briefly
review structural causal models and some important concepts from extreme value theory.
Section~\ref{sec:probaresults}
contains
a causal model for heavy-tailed distributions with {positive} coefficients.
We prove
that,
given the underlying distribution,
the causal tail coefficient
allows us
to distinguish between the causal scenarios shown in Figure~\ref{fig:6cases}.
In Section~\ref{sec:causaldisc},
we introduce an algorithm named \emph{extremal ancestral search} (EASE) that
can be applied to a matrix of causal tail coefficients and
that retrieves the causal order of the true graph, in the population case. We prove that our algorithm estimates a causal order even in the case where the causal tail coefficients are estimated empirically from data, as the sample size tends to infinity.
In Section~\ref{sec:extensions}, we first 
generalise the results of the previous sections to the case of a structural causal model with {real-valued} coefficients. To do that, we introduce a more general causal tail coefficient that is sensitive to both the upper and lower tail of the variables.
Second, we discuss the robustness properties of EASE in the presence of hidden confounders.
Third, we analyse the properties of the causal tail coefficient when the noise variables have different tail indices.
Section~\ref{sec:numerical_results} contains experiments on simulated data and real-world applications.
\ifjournal{
All proofs of the main paper are given in the Supplementary Material
\citep{gnesupp2019}.
Sections~\ref{app:regvar} to~\ref{app:financial_dynamics} are in the Supplementary Material.
}\else{
The Appendix consists of six sections.
  Section~\ref{app:regvar}	summarises important facts about regularly varying random
  variables.
  Section~\ref{app:proofs} contains the proofs of the results of the paper.
  Section~\ref{app:minimax_search} illustrates how the EASE algorithm retrieves a causal order of
  a DAG.
  Section~\ref{app:simulation_experiments} describes the settings used in the simulation study. 
  Section~\ref{app:additional_figs} contains additional figures and tables.
  Section~\ref{app:financial_dynamics} presents further results for Section~\ref{sec:financial}.
 
}\fi

\subsection{Structural causal models}\label{sec:scms}
A
linear
structural causal model, or
linear
SCM, (\citealp{Bollen1989}; \citealp[Sec.\ 1.4]{Pearl2009}) over variables $X_1, \ldots, X_p$ is a collection of~$p$ assignments
  \begin{align}\label{eq:linear-SCM}
  X_j := \sum_{k \in \pa(j)} \beta_{jk} X_k + \varepsilon_j, \quad j\in V,
  \end{align}
where $\pa(j) \subseteq V = \{1, \ldots, p\}$ and $\beta_{jk} \in \mathbb R\setminus\{0\}$,
together with a joint distribution over the noise variables $\varepsilon_1, \ldots, \varepsilon_p$.
Here, we assume that the noise variables are jointly independent and that
the induced graph $G = (V, E)$,
obtained by adding directed edges from the
parents $\pa(j)$ to $j$,
is a directed acyclic graph (DAG) with nodes $V$ and (directed) edges $E \subset V\times V$.
To ease notation, we adopt the convention to sometimes identify a node with its corresponding random variable.
To highlight the fact that $\pa(j)$ depends on a specific DAG $G$, we write $\pa(j, G)$, $j\in V$.

Structural causal models describe not only observational distributions but also interventional distributions.
An intervention on $X_j$, for example, is defined
as replacing the corresponding assignment~\eqref{eq:linear-SCM} while leaving the other equations as they were.
In practice, a causal model can be falsified
via randomised experiments \citep[see, e.g.,][Sec.\ 6.8]{Peters2017book}.

We define a directed path between node $j$ and $k$
as a sequence of distinct vertices such that successive pairs of vertices belong to the edge set $E$ of $G$.
If there is a directed path from $j$ to $k$, we say that $j$ is an ancestor of $k$ in $G$. The set of ancestors of $j$ is denoted by $\An(j, G)$, and we write $\an(j, G) = \An(j, G)\setminus\{j\}$ when we consider the ancestors of $j$ except itself.
A node $j$ that has no ancestors, i.e., $\an(j, G) = \varnothing$, is called a source node (or root node).
Given two nodes $j, k \in V$, we say that $X_j$ \emph{causes} $X_k$ if there is a directed path from $j$ to $k$ in $G$.
Furthermore, given nodes $i, j, k \in V$, we say that $X_i$ is a \emph{confounder} (or \emph{common cause}) of $X_j$ and $X_k$ if there is a directed path from node $i$ to node $j$ and $k$ that does not include $k$ and $j$, respectively. Whenever a confounder is unobserved, we say it is a  \emph{hidden} confounder or \emph{hidden} common cause.
Finally, if $\An(j, G)\cap\An(k, G) = \varnothing$,
then we say that there is \emph{no causal link} between $X_j$ and $X_k$.
A graph $G_1 = (V_1, E_1)$ is called a \emph{subgraph} of $G$ if $V_1 \subseteq V$ and $E_1 \subseteq (V_1 \times V_1)\cap E$. Recall that any subgraph of a DAG $G$ is also a DAG.
For details on graphical models, we refer to \citet{lau1996}.

\subsection{Regularly varying functions and random variables}\label{sec:evt}
A positive, measurable function $f$ is said to be
\emph{regularly varying} with index $\alpha \in \mathbb{R}$, $f\in\RV_\alpha$, if it is defined on some neighbourhood of infinity $[x_0, \infty)$, $x_0>0$, and if for all  $c > 0$,
$\lim_{x\to\infty} {f(c x)}/{f(x)} = c^\alpha$.
If $\alpha = 0$, $f$ is said to be slowly varying, $f\in \RV_0$.

A random variable $X$ is said to be \emph{regularly varying} with index $\alpha$ if
\begin{align*}
    \P(X > x) %
    \sim \ell(x)x^{-\alpha},\quad x\to\infty,
\end{align*}
for some $\ell \in \RV_0$, where for any function $f$ and $g$, we write $f \sim g$ if $f(x)/g(x) \to 1$ as $x\to\infty$.
If $X$ is regularly varying with index $\alpha$ then $cX$ is also regularly varying with the same index, for any $c > 0$.
For example, random variables with a Student's-$t$, Pareto, Cauchy, or Fr\'{e}chet distribution are regularly varying.

A characteristic property of regularly varying random variables is the  \textit{max-sum-equivalence}. The idea is that large sums of independent random variables tend to be driven by only one single large value. For this reason, the tail of the distribution of the maximum is equal to the tail of the distribution of the sum. For a rigorous formulation see
\ifjournal{Section~\ref{app:regvar} in the Supplementary Material.
}\else{Appendix~\ref{app:regvar}. }\fi
We refer to \citet[Sec.\ A3]{embrechts2013modelling} for further details on regular variation and max-sum-equivalence.

\section{The causal tail coefficient}\label{sec:probaresults}
To measure the causal effects in the extremes, we define the following parameter.
\begin{definition}\label{def:gamma}
Given two random variables $X_1$ and $X_2$, we define the \emph{causal tail coefficient}
  \begin{equation}\label{eq:gamma}
      \G_{jk} = \lim_{u\to 1^-} \E\left[F_k(X_k)  \mid F_j(X_j) > u \right],
  \end{equation}
if the limit exists, for $j, k = 1, 2$ and $j \neq k$.
\end{definition}
The coefficient $\G_{jk}$ lies between zero and one and is invariant under any marginal strictly increasing transformation
since it depends on the rescaled margins $F_j(X_j)$, for $j = 1, 2$.
 Below, we lay down the setup.

\subsection{Setup}\label{subsec:setup}
Consider a linear structural causal model (SCM) with an induced
directed acyclic graph
(DAG) $G$,
    \begin{align*}
  X_j := \sum_{k \in \pa(j, G)} \beta_{jk} X_k + \varepsilon_j, \quad j\in V,
  \end{align*}
where we assume that the coefficients $\beta_{jk}$ are strictly positive, $j, k \in V$. Let the independent noise variables $\eps_1, \ldots, \eps_p$ be real-valued and regularly varying with comparable tails, i.e., there exists a tail-index $\alpha > 0$ and $\ell \in \RV_0$ such that for all $j\in V$, there exists $c_j >0$ that satisfies
  \begin{align}\label{eq:noise}
  \P(\eps_j > x) \sim c_j \ell(x) x^{-\alpha},\quad x\to\infty.
  \end{align}
To simplify the notation, we rescale the variables $X_j$ such that $c_j = 1$, $j\in V$.
Furthermore, denote by $\beta_{k\to j}$ the sum of
distinct
weighted
directed
paths from node $k$ to node $j$, with $\beta_{j \to j} := 1$.
Since $G$ is acyclic,
we can express recursively each variable $X_j$, $j\in V$, as a \emph{weighted} sum of the noise terms $\eps_1, \dots, \eps_k$ that belong to the ancestors of $X_j$, that is,
  \begin{align}\label{eq:scm-recursive}
  X_j = \sum_{h\in \An(j, G)}\beta_{h\to j}\eps_h.
  \end{align}
The noise terms in \eqref{eq:scm-recursive} are independent and regularly varying with comparable tails as in~\eqref{eq:noise}. Therefore, by using Lemma~\ref{lemma:tailconvp}
\ifjournal{of the Supplementary Material,
}\else{of Appendix~\ref{app:regvar}, }\fi
we can write
  \begin{align*}
  \P(X_j > x) \sim \sum_{h\in\An(j, G)}\beta_{h\to j}^{\alpha}\ell(x)x^{-\alpha},\quad x\to\infty.
  \end{align*}

Any probability distribution induced by an SCM is Markov with respect to the induced DAG $G$, and thus we can read off statistical independencies
from it by $d$-separation (\citealp{Lauritzen1990}; \citealp[Sec.\ 1.2.3]{Pearl2009}).
Conversely, to infer dependencies directly from the graph, one needs to assume that the distribution is \emph{faithful} to the DAG $G$ \citep[see][Sec.\ 2.3.3]{spirtes2000causation}.
Most causal methods based on restricted SCMs, such as
LiNGAM~\citep[see][]{shimizu2006linear}, RESIT~\citep{peters2014causal}, and~\citet{peters2014identifiability}, do not assume faithfulness.
Similarly, in this work we require the milder assumption that $\beta_{j\to k}$ is non-zero if $j$ is an ancestor of $k$, i.e., $X_j$ causes $X_k$. 
This is automatically satisfied if 
the SCM has positive coefficients.
In the sequel, we refer to this model as a \textit{heavy-tailed linear SCM}.

\subsection{Causal structure and the causal tail coefficient}\label{subsec:propGamma}
In the setting of Section~\ref{subsec:setup}, the causal tail coefficient always exists and
carries information about the underlying causal structure. In particular, it can be expressed in closed form.
\begin{lemma}\label{prop:p01}
Consider a heavy-tailed linear SCM over $p$ variables. Then, for $j, k \in V$ and $j\neq k$,
  \begin{align*}
  \G_{jk} = \frac{1}{2} + \frac{1}{2}\frac{\sum_{h\in A_{jk}}\beta_{h \to j}^{\alpha}}{\sum_{h\in \An(j, G)}\beta_{h \to j}^{\alpha}},
  \end{align*}
where $A_{jk} = \An(j, G) \cap \An(k, G)$, and the sum over an empty index set equals zero.
\end{lemma}
For a proof see
\ifjournal{the Supplementary Material~\ref{proof:p01}.
}\else Appendix~\ref{proof:p01}. \fi
Lemma~\ref{prop:p01} provides a closed form expression for the causal tail coefficient $\G_{jk}$, which can be written as a sum of two terms. The first term
corresponds to the case when $X_j$ and $X_k$ are independent. The second term
is non-negative and depends on the coefficients of the SCM and the tail index
$\alpha > 0$.
By using matrix notation, it is possible to express the coefficient $\G_{jk}$ more compactly. Consider the matrix of coefficients $\matr{B}$ of the DAG $G$, where $\matr{B}_{jk} = \beta_{jk}$, $j, k \in V$,
and let $\matr{I}$ be the identity matrix. Furthermore, for any $\matr{M}\in\mathbb{R}^{p\times p}$, denote by $\matr{M}_{\alpha}$ the matrix where each entry of $\matr{M}$ is raised to the power $\alpha$. By applying the Neumann series, we obtain $\matr{H} = (\matr{I} - \matr{B})^{-1}$ where $\matr{H}_{jk} = \beta_{k \to j}$
for $j, k \in V$. Therefore, for $j,k \in V$ and $j \neq k$, we can write the causal tail coefficient as
  \begin{align}\label{eq:gamma_formula}
  \G_{jk} = \frac12 + \frac12\frac{e_j^T \matr{H}_{\alpha} e_{A_{jk}}}{e_j^T \matr{H}_{\alpha}e_{\An(j, G)}},
  \end{align}
where $e_j\in\mathbb{R}^p$ is the $j$-th standard basis vector, and $e_{C} = \sum_{j\in C} e_j \in \mathbb{R}^p$ for any set $C \subseteq \{1,\dots, p\}$.

\begin{example}
  Consider the ``diamond'' graph $G = (V, E)$ in Figure~\ref{fig:diamond_dag},
  with $V = \oneto{4}$.
\begin{figure}[!h]
\centering
\includegraphics{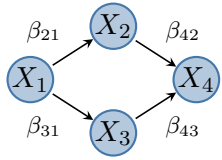}
\caption{Graphical representation of an SCM with an underlying ``diamond'' DAG $G$.}
\label{fig:diamond_dag}
\end{figure}
In this graph, for instance, it is easy to see that $\Gamma_{14} = 1$. To compute $\Gamma_{41}$, we list the weighted directed paths from the ancestors of node~4 to node~4 itself, i.e.,
  \begin{align*}
  \beta_{1\to 4} = \beta_{42}\beta_{21} + \beta_{43}\beta_{31},\
  \beta_{2\to 4} = \beta_{42},\
  \beta_{3\to 4} = \beta_{43},\
  \beta_{4\to 4} = 1.
  \end{align*}
Additionally, the set of common ancestors of node 1 and 4 is $A_{14} = \{1\}$. Putting everything together, by using Lemma~\ref{prop:p01}, or formula \eqref{eq:gamma_formula}, we obtain
  \begin{align*}
  \G_{41} = \frac{1}{2} + \frac{1}{2}\frac{\beta_{1\to 4}^{\alpha}}{\beta_{1\to 4}^{\alpha} + \beta_{2\to 4}^{\alpha} + \beta_{3\to 4}^{\alpha} + \beta_{4\to 4}^{\alpha}} < 1.
  \end{align*}
In this example, we see that it is possible to infer the causal relationship
between $X_1$ and $X_4$ because $\G_{14} > \G_{41}$.
\end{example}

Consider now a general,
heavy-tailed linear SCM over $p$ variables including $X_1$ and $X_2$,
and inducing graph $G$.
The following theorem shows that the \emph{causal tail coefficient},
which is computable from the bivariate distribution of $X_1$ and $X_2$,
see Equation~\eqref{eq:gamma}, encodes the causal relationship between the two variables.
\begin{theorem}\label{thm:6cases}
Consider a heavy-tailed linear SCM over $p$ variables including $X_1$ and $X_2$, as defined in Section~\ref{subsec:setup}.
Then, knowledge of
$\Gamma_{12}$ and
$\Gamma_{21}$ allows us to distinguish the following cases:
(a) $X_1$ causes $X_2$,
(b) $X_2$ causes $X_1$,
(c) there is no causal link between $X_1$ and $X_2$
(i.e., $\An(1, G) \cap \An(2, G) = \varnothing$),
(d) there is a node $j \not \in \{1, 2\}$, such that $X_j$ is a common cause of $X_1$ and $X_2$ and neither $X_1$ causes $X_2$ nor $X_2$ causes $X_1$.
The corresponding values for $\Gamma_{12}$ and $\Gamma_{21}$ are depicted in Table~\ref{tab:caus}.
  \begin{table}[h]
  \centering
  \caption{Summary of the possible values of $\Gamma_{12}$ and $\Gamma_{21}$ and the implications for causality.}
  \label{tab:caus}
  \begin{tabular}[t]{llll}
  \toprule
  & $\Gamma_{21} = 1$ & $\Gamma_{21} \in (1/2,1)$ & $\Gamma_{21} = 1/2$ \\
  \midrule
  $\Gamma_{12} = 1$ &   & ($a$) $X_1$ causes $X_2$ &  \\
  $\Gamma_{12} \in (1/2,1)$ & ($b$) $X_2$ causes $X_1$  & ($d$) common cause& \\
  $\Gamma_{12} = 1/2$ &   &   & ($c$) no causal link \\
  \bottomrule
  \end{tabular}
  \end{table}
\end{theorem}

For a proof see
\ifjournal{the Supplementary Material~\ref{proof:6cases}.
}\else{Appendix~\ref{proof:6cases}. }\fi
This result will also play a key role when estimating causal relationships from finitely many data. 
As a first remark, condition (a) and (b) might also include the presence of a common cause $X_j$.
As a second remark,
the empty entries in Table~\ref{tab:caus} cannot occur under the assumptions made in Section~\ref{subsec:setup}. For example, $\Gamma_{12} = \Gamma_{21} = 1$ can only happen if some variables have different tail indices. One possibility is when the \emph{cause} has a heavier tail than the \emph{effect}. Another scenario is when a common cause $X_j$, for some $j\neq 1, 2$, has heavier tails than the confounded variables $X_1$ and $X_2$. For further discussion on different tail indices see
Section~\ref{subsec:diff_tails}.

\subsection{A non-parametric estimator}\label{G_estimator} Consider a heavy-tailed linear SCM over $p$ variables including $X_1$ and $X_2$, with distributions $F_1$ and $F_2$, as described in Section~\ref{subsec:setup}.
In order to construct a non-parametric estimator of $\Gamma_{12}$ and $\Gamma_{21}$ based on independent observations $(X_{i1}, X_{i2})$, $i=1,\dots, n$, of
$(X_1,X_2)$, we define the empirical distribution function of $X_{j}$ as
    \begin{align}
    \widehat F_{j}(x) = \frac1n \sum_{i = 1}^{n}\mathbf{1}\left\{X_{ij} \leq x\right\}, \quad x \in \mathbb{R},
    \end{align}
for $j=1,2$. Denote by $g^{\gets}$ the left continuous generalised inverse
    \begin{align*}
    g^{\gets}(y) = \inf\left\{x \in \mathbb{R}:\ g(x)\geq y\right\}, \quad y\in\mathbb{R}.
    \end{align*}
In addition, let the $(n-k)$-th order statistics be denoted by $X_{(n-k),j} = \widehat F_j^{\gets}(1-k/n)$, for all $k = 0, \dots, n-1$ and $j = 1,2$, such that $X_{(1),j} \leq \dots \leq X_{(n),j}$.
Replacing $F_1$ and $F_2$ in the definition of $\G_{12}$ in \eqref{eq:gamma}
by the empirical counterparts, and the threshold $u$ by $u_n = 1 - k/n$, for some integer $0 < k \leq n-1$, we define the estimator
    \begin{align}\label{eq:gammahat}
    \widehat \Gamma_{12} = \widehat\G_{12}^{(n)} = \frac{1}{k}  \sum_{i=1}^n \widehat F_2(X_{i2}) \mathbf{1} \{X_{i1} >  X_{(n-k),1} \}.
    \end{align}
For this estimator to be consistent, a classical assumption in extreme value theory
is that the number of upper order statistics $k = k_n$ depends on the sample size $n$ such that $k_n\to\infty$ and $k_n/n \to 0$
as $n\to\infty$.
The first condition is needed to increase the effective sample size, whereas the second condition eliminates the approximation bias. The estimator $\widehat \Gamma_{21} = \widehat \Gamma_{21}^{(n)}$ is defined in an analogous way as~\eqref{eq:gammahat}.

\begin{theorem}\label{thm:Gamma_consistency}
  Let  $X_{i1}$ and $X_{i2}$, $i=1,\dots, n$, be independent copies of $X_1$
  and $X_2$, respectively, where $X_1$ and $X_2$ are two of the $p$ variables of a heavy-tailed linear SCM.

  \begin{itemize}
  \item[(A1)] Assume that the density functions $f_j = F'_j$, $j=1,2$, exist and
satisfy the \emph{von Mises' condition}
    \begin{align}\label{eq:cdfcondition}
    \lim_{x\to\infty}\frac{x f_j(x)}{1 - F_j(x)} = \frac{1}{\gamma},\quad\text{ for some } \gamma > 0.
    \end{align}
  \item[(A2)] Let $k_n\in \mathbb N$ be an intermediate sequence with
  $$k_n \to \infty \quad \text{and} \quad k_n / n \to 0, \quad n\to\infty.$$
  \end{itemize}
  Then the estimators $\widehat \Gamma_{12}$ and $\widehat \Gamma_{21}$
  are consistent, as $n\to \infty$, i.e.,
    \begin{align*}
    \widehat \Gamma_{12} \stackrel{P}{\longrightarrow} \G_{12} \text{ and }\ \widehat \Gamma_{21} \stackrel{P}{\longrightarrow} \G_{21}.
  \end{align*}
\end{theorem}
\begin{remark}
  The von Mises' condition in (A1) is a very mild assumption that is satisfied by most univariate regularly varying distributions of interest. In our case $\gamma = 1/\alpha$, where $\alpha$ is the common tail index of the noise variables.
\end{remark}
For a proof of Theorem~\ref{thm:Gamma_consistency} see
\ifjournal{the Supplementary Material~\ref{proof:Gamma_consistency}.
}\else{Appendix~\ref{proof:Gamma_consistency}. }\fi
It uses several results from tail empirical process theory \citep[e.g.,][Sec.\ 2.2]{deh2006a}.
The main challenge comes from the fact that the variables $X_1$ and $X_2$ are tail dependent, and that
the use of the empirical distribution function $\widehat F_2$ in \eqref{eq:gammahat}
introduces dependence between the terms corresponding to different observations $i=1,\dots, n$.
A related problem is studied in~\citet{cai2015}, where they derive asymptotic properties of
the empirical estimator of the expected shortfall when another dependent variable is extreme.
However, in contrast to~\citet{cai2015}, in the proof of Theorem \ref{thm:Gamma_consistency}, we work with a more explicit model and we consider the variables scaled to uniform margins, i.e., $\widehat F_j(X_j)$ instead of $X_j$, $j = 1, 2$.

\section{Causal discovery using extremes}\label{sec:causaldisc}

We would like to recover the causal information from a dataset of $p$ variables under the model specification of Section~\ref{subsec:setup}. We develop an algorithm named \emph{extremal ancestral search} (EASE) based on the \emph{causal tail coefficient} defined in~\eqref{eq:gamma}. We show that EASE can recover the causal order of the underlying graph in the population case (Section~\ref{subsec:learncausalorder}), and that it is consistent (Section~\ref{subsec:sample-prop}).

\subsection{Learning the causal order}\label{subsec:learncausalorder}
Our goal is to recover the causal order of a heavy-tailed linear SCM over $p$ variables (as defined in Section~\ref{subsec:setup}) by observing $n$ i.i.d. copies of the random vector $X\in\mathbb{R}^p$.
Given a DAG $G = (V, E)$, a permutation $\pi: \oneto{p} \rightarrow \oneto{p}$ is said to be a \emph{causal order} (or \emph{topological order}) of $G$ if $\pi(i) < \pi(j)$ for all $i$ and $j$ such that $i\in \an(j, G)$.
We denote by $\Pi_{G}$ the set of all causal orders of $G$.
For a permutation $\pi$ we sometimes use the notation $\pi = \left(\pi(1), \dots, \pi(p)\right)$.

A given causal order $\pi$ does not specify a unique DAG.
As an example, the causal order $\pi = (1, 2)$ comprises two DAGs: one where there is a directed edge between node 1 and node 2,
and one where the two nodes are unconnected.
On the other hand, there can be several causal orders for a given DAG. For example, a fully-disconnected DAG satisfies any causal order.
However, even if the causal order does not identify a unique DAG, it still conveys important information. In particular, each causal order defines a class of DAGs that agree with respect to the non-ancestral relations.
Therefore, once a causal order is available, one can estimate the complete DAG by using regularised regression methods. This idea has been exploited, e.g., by~\citet{shimizu2011directlingam} and~\citet{Buehlmann2014annals}.
In addition, \citet{Buehlmann2014annals} and~\citet{peters2015structural} argue that knowledge of a causal order is useful \emph{per se}. In fact, given a correct causal order, one can construct a fully-connected DAG that describes interventional distribution across the variables.

For any heavy-tailed linear SCM and induced DAG $G = (V, E)$,
we define the matrix
$\G\in\mathbb{R}^{p\times p}$
with entries $\G_{ij}$, the causal tail coefficients between all pairs of variables $X_i$ and $X_j$, $i,j\in V$; see Definition~\ref{def:gamma}.
Theorem~\ref{thm:6cases} tells us how the entries of $\Gamma$ encode the causal relationships between the random variables of the SCM.
To recover the causal order of the DAG $G$, we propose the algorithm named \emph{extremal ancestral search} (EASE).

\vspace{12pt}
\begin{tcolorbox}[width=\textwidth, colframe=black, colback=white, boxrule = .2mm]
\textbf{EASE algorithm}

\textsc{Input:}
A matrix $\G\in\mathbb{R}^{p\times p}$ of causal tail coefficients
related to a DAG $G = (V, E)$ with $V = \oneto{p}$.

\textsc{Returns:} Permutation of the nodes $\pi: V \to \oneto{p}$.

\begin{enumerate}
\setcounter{enumi}{-1}
\item Set $V_1 = V$.

\item \textsc{For} $s \in \{1,\dots, p\}$
\begin{enumerate}
  \item Let $M^{(s)}_i = \max_{j \in V_{s}\setminus\{i\}} \G_{ji}$, for all $i\in V_{s}$.
  \item Let $i_{s} \in \arg\min_{i\in V_{s}}M^{(s)}_i$.
  \item Set $\pi(i_{s}) = s$.
  \item Set $V_{s+1} = V_{s} \setminus \{i_{s}\}$.
\end{enumerate}
\item \textsc{Return} the permutation $\pi$.
\end{enumerate}
\textsc{Complexity:} $O(p^2)$.
\end{tcolorbox}
\vspace{12pt}

The above is a greedy algorithm that identifies root nodes of the current subgraph at each step.
In the first step, the algorithm finds a root node $i_1\in V$ as the one that minimises the score $M_{i}^{(1)} = \max_{j\neq i} \G_{ji}$, $i\in V$. In fact, by Theorem~\ref{thm:6cases}, $M_{i}^{(1)} < 1$ if and only if $i$ is a source node. Once the first node is selected, the algorithm searches for a second root node in the subgraph where $i_1$ is removed. The procedure continues until all nodes have been selected. In
\ifjournal{Section~\ref{app:minimax_search} of the Supplementary Material,
}\else{Appendix~\ref{app:minimax_search}, }\fi
one example illustrates how EASE finds a causal order for a
given DAG.

The next result states that, in the population case, the EASE
algorithm yields a correct causal order of the underlying DAG.
\begin{prop}\label{prop:minimax_oracle}
  Consider a heavy-tailed linear SCM over $p$ variables, as defined in Section~\ref{subsec:setup}, and let $G = (V, E)$ be the induced DAG. If the input $\Gamma$ is the matrix of causal tail coefficients associated with the SCM, then EASE returns a permutation $\pi$ that is a causal order of $G$.
\end{prop}
For a proof see
\ifjournal{Supplementary Material~\ref{proof:minimax_oracle}.
}\else{Appendix~\ref{proof:minimax_oracle}.}\fi

\subsection{Sample properties for the EASE algorithm}\label{subsec:sample-prop}
For finite samples, the EASE algorithm will take an estimate of the
causal coefficient matrix $\G$ as input. Based on the empirical non-parametric estimator
$\widehat\G$ and its asymptotic properties, we assess the performance of the algorithm. Let $\widehat\G\in\mathbb{R}^{p\times p}$ denote the matrix where each entry $\widehat\G_{ij}$ is defined as in~\eqref{eq:gammahat} in Section~\ref{G_estimator}, for $i, j \in V$. We say that a procedure makes a mistake when it returns a permutation $\pi\notin\Pi_G$. We derive an upper bound for the probability that EASE makes a mistake when the matrix $\G$ is estimated by $\widehat\G$.

\begin{prop}\label{prop:minimax_bound}
  Consider a heavy-tailed linear SCM over $p$ variables $X = (X_1, \ldots, X_p)$, as defined in Section~\ref{subsec:setup}, with induced DAG $G$. Let $\widehat \Gamma$ be the estimated causal coefficient matrix related to $G$. Let $\widehat\pi$ denote the permutation returned by EASE based on $\widehat \G$. Then,
    \begin{align*}
    \P\left(\widehat\pi \notin \Pi_G\right) \leq p^2 \max_{i, j\in V:i\neq j} \P\left(\left|\widehat\G_{ij} - \G_{ij}\right| > \frac{1 - \eta}{2}\right),
  \end{align*}
where $\eta = \max_{u\notin \An(v, G)} \G_{uv} < 1$.
\end{prop}
For a proof see
\ifjournal{Supplementary Material~\ref{proof:minimax_bound}.
}\else{Appendix~\ref{proof:minimax_bound}. }\fi
The bound for the probability of making a mistake in the estimated causal order
is expressed in terms of the distance between the true $\Gamma_{ij}$ and the estimated $\widehat\G_{ij}$. This bound in combination with the consistency result of Theorem~\ref{thm:Gamma_consistency}
yields the consistency of the EASE algorithm in the sample case.

\begin{cor}\label{cor_consistency}
  Let $\widehat\pi$ be the permutation computed by EASE under the assumptions of Proposition~\ref{prop:minimax_bound}. Let $k_n\in \mathbb N$ be an intermediate sequence with
  $$k_n \to \infty \quad \text{and} \quad k_n / n \to 0, \quad n\to\infty.$$
  If the von Mises' condition~\eqref{eq:cdfcondition} holds, then the EASE algorithm is consistent, i.e.,
\begin{align*}
\P\left(\widehat\pi \notin \Pi_G\right) \to 0, \quad \text{ as } n\to\infty.
\end{align*}
\end{cor}

The result above is for fixed dimension $p$. To prove consistency in a regime
where $p$ scales with the sample size $n$ we would need to establish concentration inequalities
for $\widehat \Gamma$ or asymptotic normality in Theorem \ref{thm:Gamma_consistency}. Both
would require stronger assumptions on the tails of the noise variables and a second-order
analysis in line with the proof of Theorem \ref{thm:Gamma_consistency}.

\subsection{Computational complexity}\label{subsec:compcomp}
The EASE algorithm is
based on pairwise quantities and is
therefore computationally efficient.
To estimate the matrix $\widehat \G$ of causal tail coefficients, which is the input for EASE, first, we need to rank the $n$ observations for each of the variables, with a computational complexity of $O(pn\log n)$. Then we compute the coefficients $\widehat\G_{ij}$ for each pair $i, j\in V$, with a computational complexity of $O(k_n p^2)$.
The computational complexity of EASE grows with the square of the number of variables, i.e., $O(p^2)$.
The overall computational complexity of estimating the matrix $\widehat \G$ and running the EASE algorithm is therefore $O\left(\max(pn\log n, k_n p^2)\right)$.

\section{Extensions}\label{sec:extensions}

\subsection{Real-valued coefficients}\label{subsec:realvalbeta}
Until now, we have worked with a heavy-tailed linear SCM with positive coefficients (see Section~\ref{subsec:setup} for a detailed explanation of the model). In the current section,
we relax this assumption and let the coefficients of the SCM be real-valued, i.e., $\beta_{jk}\in \mathbb{R}$, $j, k\in V$.
Additionally, we assume that $\beta_{j \to k}$ is non-zero, if $j$ is an ancestor of $k$.
Given that the coefficients are real-valued, we want to consider both the upper and the lower tails of the variables. We assume that the noise variables $\eps_1, \dots, \eps_p$, of the SCM have comparable upper and lower tails, that is, as $x\to\infty$
  \begin{align*}
  \P(\eps_j > x) \sim c_j^{+}\ell(x)x^{-\alpha},\quad
  \P(\eps_j < -x) \sim c_j^{-}\ell(x)x^{-\alpha},
  \end{align*}
where $c_j^{+}, c_j^{-} > 0$, $j\in V$ and $\ell \in \RV_0$.
Furthermore, we define a causal tail coefficient that is sensitive to both tails as
  \begin{align}\label{eq:psi}
  \begin{split}
  \Psi_{jk}
  = &\ \lim_{u\to 1^-}\E\left[\sigma(F_k(X_k)) \mid \sigma(F_j(X_j)) > u\right], \quad j, k\in V,
  \end{split}
  \end{align}
if the limit exists, where $\sigma: x\mapsto |2x -1|$. Since $F_j(X_j) \sim \text{Unif}[0, 1]$, $j\in V$,
we can rewrite~\eqref{eq:psi} as
  \begin{align}\label{eq:psi-decomp}
  \begin{split}
  \Psi_{jk}
  = &\ \lim_{u\to 1^-}\frac{1}{2}\E\left[\sigma(F_k(X_k)) \mid F_j(X_j) > u\right]\\
  &\ + \lim_{u\to 0^+}\frac{1}{2}\E\left[\sigma(F_k(X_k)) \mid F_j(X_j) < u\right]\\
  = &\ \Psi_{jk}^{+} + \Psi_{jk}^{-},
  \end{split}
  \end{align}
where the first and second terms correspond to the cases where $X_j$ is extremely large and extremely small, respectively.

In the current setting, the coefficient defined in~\eqref{eq:psi}
always exists, and it has a closed form expression that
encodes causal relationships between the variables.
\begin{lemma}\label{lemma:psi-coeff}
Consider a heavy-tailed SCM over $p$ variables, where the coefficient $\beta_{jk}\in\mathbb{R}$, $j, k\in V$.
Assume that $\beta_{j \to k}\neq 0$  if $j$ is an ancestor of $k$.
Then, for $j, k\in V$ and $j\neq k$,
  \begin{align*}
  \Psi_{jk} = \frac{1}{2}
  + \frac{1}{4} \frac{\sum_{h\in A_{jk}} c_{hj}^{+}|\beta_{h\to j}|^{\alpha}}{\sum_{h\in \An(j, G)} c_{hj}^{+}|\beta_{h\to j}|^{\alpha}}
  + \frac{1}{4} \frac{\sum_{h\in A_{jk}} c_{hj}^{-}|\beta_{h\to j}|^{\alpha}}{\sum_{h\in \An(j, G)} c_{hj}^{-}|\beta_{h\to j}|^{\alpha}},
  \end{align*}
where $A_{jk} = \An(j, G) \cap \An(k, G)$, and
  \begin{align}\label{eq:sign-func}
  c_{hj}^{+} =
  \begin{cases}
  c_h^+, \quad \beta_{h\to j} > 0,\\
  c_h^-, \quad \beta_{h\to j} < 0,
  \end{cases}
  \quad\quad
  c_{hj}^{-} =
  \begin{cases}
  c_h^-, \quad \beta_{h\to j} > 0,\\
  c_h^+, \quad \beta_{h\to j} < 0.
  \end{cases}
  \end{align}
\end{lemma}
A proof is provided in
\ifjournal{the Supplementary Material~\ref{proof:psi-coeff}.
}\else{Appendix~\ref{proof:psi-coeff}. }\fi
The interpretation of the result is
as follows. The baseline of the coefficient is $1/2$, which can be checked to be the value of $\Psi_{jk}$ when two variables are independent. The other two terms account for the equally weighted contribution from the lower and upper tail, respectively. The result stated in Lemma~\ref{lemma:psi-coeff} allows us to extend Theorem~\ref{thm:6cases} to the more general setting where the heavy-tailed SCM has
coefficients $\beta_{jk}\in\mathbb{R}$, $j, k\in V$.

\begin{theorem}\label{thm:6cases-psi}
Consider a heavy-tailed linear SCM over $p$ variables including $X_1$ and $X_2$, and assume that $\beta_{jk}\in\mathbb{R}$, $j, k\in V$.
In addition, assume that $\beta_{j \to k}\neq 0$ if $j$ is an ancestor of $k$.
Then, knowledge of
$\Psi_{12}$ and
$\Psi_{21}$ allows us to distinguish the following cases:
(a) $X_1$ causes $X_2$,
(b) $X_2$ causes $X_1$,
(c) there is no
causal link between $X_1$ and $X_2$,
(d) there is a node $j \not \in \{1, 2\}$, such that $X_j$ is a common cause of $X_1$ and $X_2$ and neither $X_1$ causes $X_2$ nor $X_2$ causes $X_1$.
The corresponding values for $\Psi_{12}$ and $\Psi_{21}$ are shown in Table~\ref{tab:caus_psi}.
  \begin{table}[h]
  \centering
  \caption{Summary of the possible values of $\Psi_{12}$ and $\Psi_{21}$ and the implications for causality.}
  \label{tab:caus_psi}
  \begin{tabular}[t]{llll}
  \toprule
  & $\Psi_{21} = 1$ & $\Psi_{21} \in (1/2,1)$ & $\Psi_{21} = 1/2$ \\
  \midrule
  $\Psi_{12} = 1$ &   & ($a$) $X_1$ causes $X_2$ &  \\
  $\Psi_{12} \in (1/2,1)$ & ($b$) $X_2$ causes $X_1$  & ($d$) common cause& \\
  $\Psi_{12} = 1/2$ &   &   & ($c$) no causal link \\
  \bottomrule
  \end{tabular}
  \end{table}
\end{theorem}
The proof is identical to the proof of Theorem~\ref{thm:6cases}, replacing $\Gamma_{ij}$ with $\Psi_{ij}$ and by referring to Lemma~\ref{lemma:psi-coeff} instead of Lemma~\ref{prop:p01}.
Moreover, as in Theorem~\ref{thm:6cases}, condition (a) and (b) can also include the presence of a common cause $X_j$.
Theorem~\ref{thm:6cases-psi} implies that if we run the EASE algorithm based on the matrix $\Psi\in\mathbb{R}^{p\times p}$, containing the pairwise $\Psi_{ij}$, $i, j \in V$, then we retrieve a causal order of the underlying DAG. This is the analogue to Proposition~\ref{prop:minimax_oracle} for heavy-tailed linear SCM with real-valued coefficients.

We define an empirical estimator $\widehat \Psi_{ij}$ of $\Psi_{ij}$ in a similar fashion as the estimator $\widehat \Gamma_{ij}$ in~\eqref{eq:gammahat}. The proof of Lemma~\ref{lemma:psi-coeff} shows that the coefficient $\Psi_{ij}$ can be decomposed in the same way as  $\Gamma_{ij}$ in the proof of Lemma~\ref{prop:p01}. Therefore, following the lines of the proof of Theorem~\ref{thm:Gamma_consistency} with some minor modifications, we obtain the consistency $\widehat \Psi_{ij} \stackrel{P}{\longrightarrow} \Psi_{ij}$ as $n\to \infty$ for any intermediate sequence $k_n\to \infty$ and $k_n/n \to 0$, and under the assumption of the von Mises' condition for both the upper and the lower tail of $X_j$, $j\in V$.

We can then estimate a permutation $\widehat\pi$ by the EASE algorithm based on the matrix $\widehat\Psi \in \mathbb R^{p\times p}$ that contains the estimators $\widehat \Psi_{ij}$, $i,j \in V$, as entries. For this permutation we obtain the same bound for the probability that EASE makes a mistake as shown in Proposition~\ref{prop:minimax_bound} by replacing $\widehat \Gamma_{ij}$ and $\Gamma_{ij}$ by $\widehat \Psi_{ij}$ and $\Psi_{ij}$, respectively. This together with the consistency of $\widehat \Psi_{ij}$ yields the following result.

\begin{cor}\label{cor_consistency_psi}
Assume the general setup of the heavy-tailed linear SCM with real-valued coefficients of this section. Let $\widehat\pi$ be the permutation computed by EASE based on the matrix $\widehat\Psi$. Assume the von Mises' condition for the upper and the lower tail of $X_j$ and let $k_n\in \mathbb N$ be an intermediate sequence with
  $$k_n \to \infty \quad \text{and} \quad k_n / n \to 0, \quad n\to\infty.$$
  Then, the EASE algorithm is consistent, i.e.,
\begin{align*}
\P\left(\widehat\pi \notin \Pi_G\right) \to 0, \quad \text{ as } n\to\infty.
\end{align*}
\end{cor}

\subsection{Presence of hidden confounders}\label{subsec:confounders}

A frequent assumption in causality is that one can observe all the relevant variables.
However, in many real-world situations, it is hard, if not impossible, to do so.
When some of the hidden variables are confounders (i.e., common causes), the causal inference process might be compromised.
Therefore, an attractive property of a causal inference algorithm involves its robustness to hidden confounders.
In this section, we show that EASE is capable of dealing with hidden common causes and, under certain assumptions, it recovers the causal order of the observed graph both in the population and in the asymptotic case.

Consider a heavy-tailed
linear SCM with \emph{real-valued} coefficients, as defined in Section~\ref{subsec:realvalbeta}, consisting of both observed and hidden variables.
This SCM induces a DAG $G = (V, E)$, with $V = V_O \cup V_H$, $V_O \cap V_H =\varnothing$, where $V_O$ ($V_H$) denotes the set of nodes corresponding to the observed (hidden) variables. Our goal is to recover a causal order for the subset of the observed variables $X_j$, $j \in V_O$.
In particular, we say that the EASE algorithm recovers a causal order $\pi$ over the observed variables
if
  \begin{align}\label{eq:caus-ord-observed}
   \pi(i) < \pi(j) \implies j\notin \an(i, G),\quad \text{for all } i, j \in V_O.
  \end{align}
In fact, the results of the previous sections hold even in the presence of
hidden confounders.

Regarding the population properties,
Theorem~\ref{thm:6cases} and~\ref{thm:6cases-psi} still
apply:
they state that the causal tail coefficients $\Gamma$ and $\Psi$ reflect
the causal relationships between pairs of variables 
without taking into account other variables, e.g., by conditioning.
In addition, the result of Proposition~\ref{prop:minimax_oracle},
and the corresponding
extension in Section~\ref{subsec:realvalbeta}, are also valid.
The proof of Proposition~\ref{prop:minimax_oracle} depends only on the
assumption that the input matrix
contains the pairwise causal effects between the variables.
Therefore, if we use matrix $\Gamma$ (or $\Psi$) as input for the EASE algorithm,
we recover a causal order $\pi$ that satisfies~\eqref{eq:caus-ord-observed}.

Regarding the asymptotic properties, 
$\hat\Gamma$ and $\hat\Psi$ are consistent even in the presence of hidden common causes:
in the proof of Theorem~\ref{thm:Gamma_consistency}, 
the other variables do not appear.
In addition, we can still
find an upper bound for
the probability that the EASE algorithm makes a mistake.
To do so, one needs to adjust the proof of Proposition~\ref{prop:minimax_bound} by replacing
the full DAG $G$ with the subgraph $G_O = (V_O, E_O)$ containing only the observed variables, where $E_O = E \cap (V_O\times V_O)$.
Combining the two previous arguments, it follows that
Corollary~\ref{cor_consistency} and~\ref{cor_consistency_psi} hold even
in the presence of hidden confounders.

The ability to deal with hidden confounders is a property that, in general, is not shared by all methods in causality. For example, the PC algorithm~\citep[Sec.\ 5.4.2]{spirtes2000causation} might retrieve a Markov equivalence class that contains DAGs with a wrong causal order
if some of the variables are not included in the analysis.
Similarly, the standard version of the LiNGAM algorithm \citep{shimizu2006linear} might produce a wrong DAG in the presence of hidden common causes.
\citet{hoyer2008estimation},  \citet{entner2010discovering}, and~\citet{tashiro2014parcelingam} proposed extensions of LiNGAM that deal with hidden variables. While \citet{entner2010discovering}, and~\citet{tashiro2014parcelingam} show good performance in practice, all three methods suffer from some drawbacks. For example, the LiNGAM version of~\citet{hoyer2008estimation} requires \emph{a priori} the number of hidden variables in the SCM
(or needs to estimate it from data). 
The main limitation of~\citet{entner2010discovering} is that it recovers causal information only for subsets of variables that are not affected by hidden confounders. For some non-ancestral
graphs, the method by~\citet{tashiro2014parcelingam} does not identify all ancestral relationships~\citep{wang2020causal}.
In addition, both the work of~\citet{hoyer2008estimation} and~\citet{tashiro2014parcelingam} are computationally
intensive, with the latter showing a computational time that grows exponentially with the sample size and the number of observed variables.
Among the constraint-based methods, \citet[Sec.\ 6.7]{spirtes2000causation} proposed the FCI method, which is an extension to the PC algorithm that deals with arbitrarily many hidden confounders and produces a partial ancestral graph \citep[see][]{zhang2008causal}. Due to the high number of independence tests, the FCI algorithm can be slow when the number of variables is large. For this reason, \citet{claassen2013learning} proposed the FCI+ algorithm,
a faster version of FCI that is consistent in sparse high-dimensional settings with arbitrarily many hidden variables.
In general, FCI type algorithms produce an equivalence class of graphs. They are not guaranteed to recover the causal order of the variables.

Compared to the methods mentioned above, our algorithm has the advantage of being computationally fast,
and being able to produce a causal order
without assumptions on the number of hidden variables and the sparsity of the true underlying DAG.

\subsection{Noise variables with different tails}\label{subsec:diff_tails}

We have so far considered the case where the noise variables of a given SCM
share the same tail coefficient $\alpha > 0$ and the same slowly varying function $\ell$.
Consider now a heavy-tailed SCM over $p$ variables, as defined in
Section~\ref{subsec:setup}, with the difference that the noise variables have possibly different tail
indices $\alpha_1, \dots, \alpha_p > 0$ and slowly varying functions $\ell_1,\dots, \ell_p\in \RV_0$, i.e., for $j = 1, \dots, p$,
  \begin{align*}
  \P(\eps_j > x) \sim \ell_j(x)x^{-\alpha_j}, \quad x\to \infty.
  \end{align*}
We say that $\eps_j$ has heavier (upper) tail than $\eps_k$ if either $0< \alpha_j < \alpha_k$,  or $\alpha_j = \alpha_k$ and $\ell_j(x) / \ell_k(x) \to \infty$ as $x\to \infty$. 
Denote by $G = (V, E)$ the DAG induced by the SCM.
With similar arguments to the proof of Lemma~\ref{prop:p01}, the causal tail coefficient for $j, k \in V$
can then be expressed as
  \begin{align}\label{eq:diff_tails}
  \G_{jk} = \frac{1}{2} + \frac{1}{2} \lim_{x \to \infty}
   \frac{\sum_{h\in A_{jk}} \beta_{h\to j}^{\alpha_h}
   \P(\eps_h > x)}{\sum_{h\in\An(j, G)} \beta_{h\to j}^{\alpha_h}\P(\eps_h > x)},
  \end{align}
where $A_{jk} = \An(j, G) \cap \An(k, G)$, and the sum over an empty index set
equals zero.
From~\eqref{eq:diff_tails}, we can study the different constellations of $X_j$ and $X_k$ and the corresponding values of $\G_{jk}$. 
We obtain the following three statements.
\begin{enumerate}
\item If $X_j$ and $X_k$ are independent, the
causal tail coefficient satisfies, as before, $\G_{jk} = 1/2$.
\item
  If $X_j$ is an ancestor of $X_k$ then, as before, $\G_{jk} = 1$.
\item
  In all other scenarios, it holds that $\G_{jk} < 1$ as long as the noise variables $\eps_h$, $h\in A_{jk}$, of the common ancestors of $X_j$ and $X_k$ 
  have tails that are lighter than (or as light as) the one of $\eps_j$.
  On the other hand, if there is some common ancestor of $X_j$ and $X_k$ 
for which the noise variable's tail is heavier than the one of $\eps_j$, then $\G_{jk} = 1$.
\end{enumerate}
These statements help to understand in which cases the values of $\G_{jk}$ indicate a correct causal relation.
\begin{example}\label{ex_2}
  Suppose that $X_j$ is an ancestor of $X_k$, and there is possibly a common ancestor $X_0$ (which can also be a hidden confounder). Since $\Gamma_{jk}=1$, we will never 
  mistakenly detect the existence of a
  causal effect from $X_k$ to $X_j$. 
If either $X_j$ or $X_0$ has a heavier tail than $X_k$, then $\G_{kj}=1$ and we cannot detect the causal effect from $X_j$ to $X_k$.
  \end{example}

\begin{example}\label{ex_com_an}
  Suppose neither $X_j$ causes $X_k$ nor $X_k$ causes $X_j$, 
and $X_0$ is a common ancestor of $X_j$ and $X_k$.
   If $X_0$ has a tail that is lighter than (or as light as)
  the one of $X_k$, but heavier than the one of $X_j$,
  then $\G_{jk} = 1 > \G_{kj}$. Therefore, the causal tail coefficient indicates a wrong causal effect from $X_j$ to $X_k$.
\end{example}

Whenever there exists a causal effect
between two variables, we can, at worst, fail 
to detect it (that is, the causal tail coefficient does not indicate a causal effect in the wrong direction).
When there is no causal connection between
two variables, the causal tail coefficient might indicate a wrong causal effect.
However, this does not affect the 
correctness of the
EASE algorithm, on the population level. 
Indeed, if $\G_{jk} = 1 > \G_{kj}$, for $j, k \in V$, there are two possibilities.
If $X_j$ is an ancestor of $X_k$, then the algorithm correctly chooses $j$ before $k$.
If $X_j$ and $X_k$ share a common ancestor, but none of them is causing the other (see Example~\ref{ex_com_an}), then
any permutation of $j$ and $k$ yields a valid causal order.
Example~\ref{ex_2} shows that EASE could make mistakes when 
$X_j$ is an ancestor of $X_k$ and
$\G_{jk} = \G_{kj} = 1$, 
since
the causal tail coefficient does not indicate any causal effect.
In this case, one could remove one variable at a time to obtain a subset $\tilde V \subset V$ that satisfies $\G_{hm} < 1$ or $\G_{mh} < 1$ for each $h,m\in \tilde V$. 
By applying the EASE algorithm to the subset of the remaining variables, one would recover a correct causal order on such subset. 

For simplicity we only considered the causal tail coefficient $\G$,
but similar conclusions hold for $\Psi$.

\ifjournal{}\fi
\section{Numerical results}\label{sec:numerical_results}

\subsection{Simulation study}\label{sec:simulation_study}
We assess the performance of EASE in estimating a causal order of a graph
induced by a heavy-tailed SCM. We simulate the SCMs with real-valued
coefficients and different numbers of variables $p$ and samples $n$.
The noise variables have Student's~$t$ distributions with different degrees of
freedom $\alpha$ and we consider four different settings, including unobserved
confounders and model misspecification;
see
\ifjournal{Section~\ref{app:simulation_experiments} of the Supplementary Material
}\else{Appendix~\ref{app:simulation_experiments} }\fi
for details.
Since the coefficients in the SCM are real-valued, we use the causal tail
coefficient $\Psi$ defined in Section~\ref{subsec:realvalbeta} for our EASE algorithm.
Our code is available as an \texttt{R} package
at \url{https://github.com/nicolagnecco/causalXtreme}.
Scripts generating all our figures and results can be found at the same url.

\subsubsection{Competing methods and evaluation metric}
We compare our algorithm to three well-established methods in causality,
the Rank PC algorithm~\citep{harris2013pc}, ICA-LiNGAM~\citep{shimizu2006linear},
and Pairwise LiNGAM~\citep{hyvarinen2013pairwise}.

The classic PC algorithm \citep[Sec.\ 5.4.2]{spirtes2000causation} belongs to the class of constraint-based methods for causal discovery. It estimates the Markov equivalence class of a DAG, encoded as a completed partially directed acyclic graph (CPDAG).
The PC algorithm retrieves a CPDAG by performing conditional independence tests
on the variables.
The Rank PC algorithm, proposed by~\citet{harris2013pc}, is an extension of the
PC algorithm and uses the rank-based Spearman correlation
to perform the independence tests. This modification ensures that the method is
more robust to non-Gaussian data.

The algorithms that fit our problem best are ICA-LiNGAM and Pairwise LiNGAM.
ICA-LiNGAM, proposed by~\citet{shimizu2006linear},
leverages the results of independent component analysis (ICA) \citep{comon1994independent} to estimate the DAG of a linear SCM
under the only assumption that the noise is non-Gaussian.
Pairwise LiNGAM, proposed by~\citet{hyvarinen2013pairwise},
is a likelihood-ratio-based method to identify the exogenous variables within
the DirectLiNGAM framework.
DirectLiNGAM, introduced by~\citet{shimizu2011directlingam},
is an algorithm based on two iterative steps, namely, finding an
exogenous variable (i.e., a node in the DAG with no parents), and regressing this variable out of all the others.
In this simulation study, we let ICA-LiNGAM and Pairwise LiNGAM return only a causal order (and not a complete DAG structure), in order to make a fair comparison with EASE.

The algorithms return different types of causal information. On the one hand,
EASE, ICA-LiNGAM, and Pairwise LiNGAM estimate a causal order.
On the other hand,
the Rank PC algorithm computes a CPDAG
that represents a Markov equivalence class of DAGs.
Therefore, when it comes to evaluating the performance of the algorithms,
it becomes crucial to use a measure that is meaningful for all of them.
We choose the structural intervention distance (SID) proposed by \citet{peters2015structural}. The SID takes as input either a pair of DAGs or a DAG and a CPDAG and returns the number of falsely inferred interventional distributions \citep[Definition 3]{peters2015structural}. We standardise the SID to lie between zero and one.
For each method,  we compute the distance between the simulated DAG, i.e., the ground truth, and the estimated DAG or CPDAG.
An estimated causal order $\widehat \pi$ corresponds to a fully
connected DAG $G = (V, E)$, where $(i, j) \in E$ if
$\widehat\pi(i) < \widehat\pi(j)$.
As a caveat, we slightly adapt the SID in the case of hidden confounders (see Setting~2 of our simulations), since it is not designed to work in such a situation.

\subsubsection{Results}\label{subsec:sim_results}
In this simulation experiment we use the implementation of the Rank PC, and ICA-LiNGAM algorithm developed by \citet{kalisch2012causal}.
We implemented Pairwise LiNGAM in \texttt{C++} and included it in our software package.

Regarding the hyperparameter settings, for the Rank PC algorithm, we perform
a conditional independence test based on Spearman's correlation coefficient, as proposed by~\citet{harris2013pc}, and we set the level of the independence tests to 0.0005.

Concerning the choice of the number of exceedances $k_n$ in the EASE algorithm, we perform a small preliminary simulation. Figure~\ref{fig:k-robustness} shows the SID of EASE for $k_n =  \lfloor n^{\nu}\rfloor$ and different fractional exponents $\nu>0$.
The best fractional exponent in Figure~\ref{fig:k-robustness} seems to depend on the tail heaviness of the noise variables, and in particular it appears to be smaller for larger values of the degree of freedom $\alpha$ of the Student's $t$ distribution. Our estimators $\widehat \Gamma_{ij}$ and $\widehat \Psi_{ij}$ are similar in construction to Hill's estimator~\citep{hill1975simple}. For the latter, the optimal number $k_n^*$ of exceedances depends on the tail index and an index related to a second-order condition; see Section~3.2 in \cite{deh2006a} for details. For the Student's $t$ distribution with $\alpha$ degrees of freedom it can be shown that $k_n^* \sim C_\alpha n^{1/(\alpha+1)}$, where $C_\alpha >0$ is a constant. This intuitive explanation coincides well with the optimal fractional exponents in Figure~\ref{fig:k-robustness}.
In the sequel, we choose $k_n = \lfloor n^{0.4}\rfloor$ because it lies within
the best range for the fractional exponent.
This result also agrees with the assumptions of Theorem~\ref{thm:Gamma_consistency}, where $k_n \to \infty$ and  $k_n / n \to 0$, as $n\to\infty$.

Regarding the simulation settings, we let $n$ denote the number of observations, $p$ the number of variables, and $\alpha > 0$ the tail index of the simulated distribution.
For each combination of $n \in\{500, 1000, 10000\}$, $p \in \{4, 7, 10, 15, 20, 30,  50\}$ and $\alpha \in \{1.5, 2.5, 3.5\}$ we simulate 50 random SCMs under four different settings. The simulated data is independent of the data used to choose the best fractional exponent of $k_n$ (see Figure~\ref{fig:k-robustness}).
The first setting corresponds to linear SCMs with real-valued coefficients described in Section~\ref{subsec:realvalbeta}. In the second setting, we introduce hidden confounders. The third setting corresponds to nonlinear SCMs. In the fourth setting, we first generate linear SCMs and then transform each variable to uniform margins.
Further details on the generation of the SCMs are in
\ifjournal{Section~\ref{app:simulation_experiments} of the Supplementary Material.
}\else{Appendix~\ref{app:simulation_experiments}. }\fi
For each simulation and setting we evaluate the performance of EASE, ICA-LiNGAM, Pairwise LiNGAM, and Rank PC algorithm with the SID.
As a baseline, in each simulation, we also compute the SID of a randomly generated DAG, where we randomly choose the causal order, the sparsity and the edges of the graph.

Figure~\ref{fig:sid_plot} displays the results of the simulations when the tail index $\alpha = 1.5$.
We can observe that EASE is quite robust across the four different settings. We explain this finding as follows. In the presence of hidden confounders (Setting~2), EASE can retrieve a correct causal order, asymptotically. Furthermore, the nonlinear setting used in this simulation (Setting~3) is such that the relationships between the variables are kept linear in the tails. Therefore, our algorithm is only moderately affected by this model misspecification. Finally, EASE is not affected by the transformation to uniform margins (Setting~4) because the causal tail coefficient $\Psi$ is invariant under any strictly monotone increasing transformation.

Compared to the other methods, we observe that EASE performs better than Rank
PC across all settings, and better than ICA-LiNGAM in Setting~2 and~4.
Pairwise LiNGAM is overall the best performing method, except in Setting~4.
Also, Pairwise LiNGAM is less affected by misspecifications in the bulk of the
data distribution (Setting~3), compared to ICA-LiNGAM. One reason is that
Pairwise LiNGAM relies on ordinary least square regression that is sensitive
to high-leverage points.
In this particular setting, ICA-LiNGAM and Pairwise LiNGAM
gain in robustness if we discard the data in the bulk of the distribution; see Table~\ref{tab:lingam_experiment}
\ifjournal{in the Supplementary Material.
}\else{in Appendix~\ref{app:additional_figs}. }\fi
Furthermore, both ICA-LiNGAM and Pairwise LiNGAM are the algorithms with the best convergence for high dimension $p$, as $n$ increases.
This result is not surprising because EASE uses only the $k_n < n$ upper order statistics to recover the causal structure.
In addition, we notice that Pairwise LiNGAM outperforms ICA-LiNGAM, in agreement to the findings of~\citet{hyvarinen2013pairwise}.
Regarding the Rank PC algorithm, we can see that it is quite stable under different settings, but it performs only marginally better than the random method.
Moreover, in Figure~\ref{fig:rankpc}
\ifjournal{of the Supplementary Material,
}\else{in Appendix~\ref{app:additional_figs}, }\fi
we observe that the performance of the Rank PC algorithm is almost constant for  significance levels of the independence test between $5\cdot 10^{-4}$ and $0.5$.
The results do not change when we consider the standard PC algorithm, which is based on partial correlation as a conditional independence test. 
The results for tail indices $\alpha = 2.5, 3.5$ are in Figures~\ref{fig:sid_plot_alpha_25} and~\ref{fig:sid_plot_alpha_35}
\ifjournal{in the Supplementary Material.
}\else{in Appendix~\ref{app:additional_figs}. }\fi
Increasing values of $\alpha$ correspond to lighter tails, and we observe that
it becomes more challenging for EASE to recover the correct causal order.

In addition to the competitive performance in the simulations, a further advantage of EASE is its computational efficiency. The algorithm performs computations only on the tails of the dataset
and relies on simple non-parametric estimators of the causal tail coefficient; see Section~\ref{subsec:compcomp}.
Figure~\ref{fig:time} 
\ifjournal{in the Supplementary Material
}\else{in Appendix~\ref{app:additional_figs} }\fi
 shows that EASE can be up to two orders of
magnitude faster than the other methods.

\begin{figure}[h!]
\centering
\includegraphics[scale=0.55]{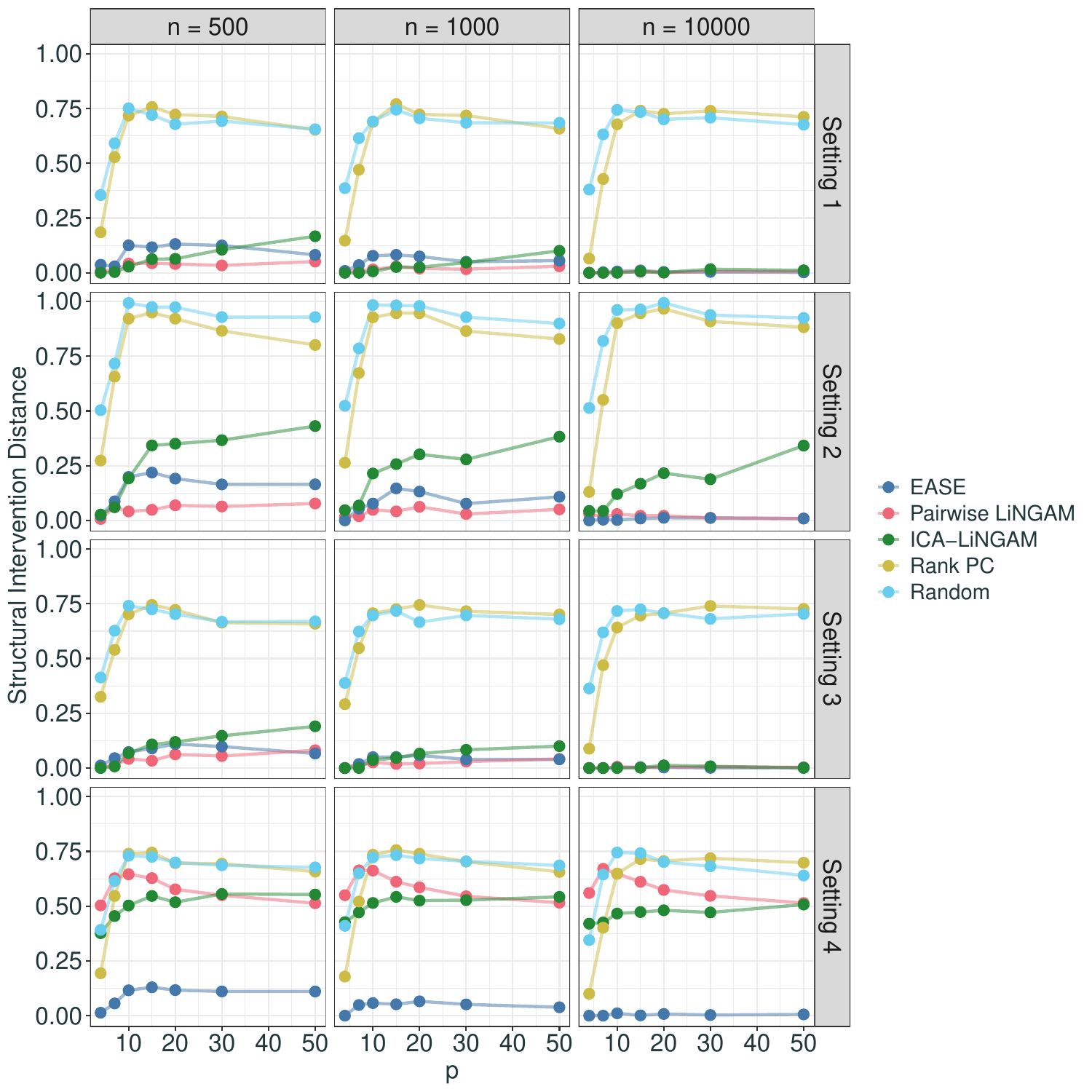}
\caption{The figure refers to Section~\ref{subsec:sim_results}. It shows the SID averaged over 50 simulations, for each method, setting, sample size $n$ and dimension $p$, when the tail index is $\alpha = 1.5$.
Each row of the figure corresponds to one setting. In order, the settings are: (1)~Linear SCM; (2)~Linear SCM with hidden confounders; (3)~Nonlinear SCM; (4)~Linear SCM where each variable is transformed to a uniform margin.}
\label{fig:sid_plot}
\end{figure}

\subsection{Financial application}\label{sec:financial}
In general, one cannot easily reason about causality in financial markets.
Several factors influence financial returns, and most of them are unobserved.
In addition, the effect of these factors varies in time.
However, under particular circumstances, it is possible to conjecture the existence
of a specific causal relationship, with a reasonable
degree of confidence.
For example, in the Swiss financial market, one can argue that very large (both
positive and negative) changes to the Euro Swiss franc exchange rate (EURCHF)
induce
changes in the Swiss Market Index (SMI), the main stock index in Switzerland.
This is due to  multiple reasons such as the multinational nature and the high export dependency of the Swiss economy.
Consider, for instance, the 
  decision of the
Swiss National Bank (SNB) to discontinue
the minimum exchange rate between Swiss franc and Euro, on January 15, 2015.
This event can be deemed as a \emph{large intervention} 
  with a plausible causal
interpretation (in the spirit of \citet[Sec.\ 8.7]{cox1996multivariate}).
Following the SNB decision, the EURCHF plummeted
more than 30 standard deviations, and all the stocks included in the
SMI dropped in value on the same day.

For this reason, we consider the returns of the Euro Swiss franc exchange rate
(EURCHF) and the three largest Swiss stocks in terms of market capitalisation,
namely, Nestl\'{e} (NESN), Novartis (NOVN) and Roche (ROG).
We choose to analyse three individual stocks instead of the SMI for three reasons.
First, we deem it more appropriate to test our assumptions on more than
two variables.
Moreover, the three stocks make up 50\% of the SMI composition, and thus they are
representative of the index itself.
Furthermore, all three companies are multinational corporations with a homogeneous
exposure to foreign markets and a negligible fraction of revenues coming from the
Swiss market \citep[see][]{nestle2019, novartis2019, roche2019}.
The last point suggests that the effect (if any) of EURCHF on these stocks
does not depend on the idiosyncrasies of each firm.

The dataset consists of daily returns spanning from January 2005 to September 2019 and includes $n = 3832$ observations.
The goal is to assess whether EASE can retrieve a correct causal order for the set
of four variables.
As ground truth, we conjecture that large changes in EURCHF (both positive and negative)
will cause large changes in the stock returns, but not \emph{vice versa}. Figure~\ref{fig:finance_dag} shows the causal structure corresponding to our hypothesis.

\begin{figure}[H]
\centering
\includegraphics[scale=1]{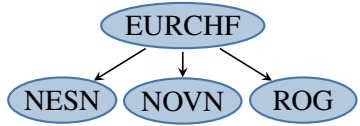}
\caption{A DAG representing a plausible causal structure among the daily returns of Euro Swiss franc exchange rate (EURCHF), Nestl\'{e} (NESN), Novartis (NOVN) and Roche (ROG).}
\label{fig:finance_dag}
\end{figure}

Before running the EASE algorithm, we assess the tail behaviour of each variable by estimating the shape parameter $\xi$ of a generalised Pareto distribution \citep[see][Sec.\ 3.4]{embrechts2013modelling} on the threshold data. Recall that $\xi$ is the reciprocal of the tail index $\alpha = 1/\xi$, if $\xi > 0$. For each variable and each tail (upper and lower) we estimate the $\xi$ parameter using 200 observations, corresponding to the 95\%-quantile, approximately.
The estimated parameters and their standard errors for the upper tails are 0.31 (0.08) for EURCHF, 0.25 (0.08) for NESN, 0.16 (0.07) for NOVN, and 0.25 (0.09) for ROG. Regarding the lower tail, the estimated parameters and their standard errors are 0.27 (0.08) for EURCHF, 0.12 (0.08) for NESN, 0.17 (0.08) for NOVN, and 0.23 (0.10) for ROG.
By adding and subtracting two standard errors from each point estimate of $\xi$,
we observe that the lower shape parameter of Nestl\'{e} and Novartis is not significantly
different from zero. Moreover, the shape parameters for the Euro Swiss franc
exchange rate and for Roche are significantly different from zero in both tails.
In addition, since the confidence intervals of all estimates are overlapping, the assumption of a common shape parameter seems reasonable. It seems however that the returns of Nestl\'{e} and Novartis have slightly lighter tails compared to the other two variables.

With the goal of recovering a causal order
with EASE,
first we estimate the $\widehat\Psi$ matrix from the full dataset,
by setting the number of exceedances to $k = 10$ (this corresponds to
$\lfloor n ^{0.3}\rfloor$, approximately).
We run the EASE algorithm on the matrix $\widehat\Psi$, and we obtain
the causal order $\widehat\pi^{-1}$ = (EURCHF, NOVN, ROG, NESN); this
agrees with the proposed ground truth of Figure~\ref{fig:finance_dag}.
As a comparison, also ICA-LiNGAM and Pairwise LiNGAM recover a causal order
that agrees with our hypothesis.

Since our results are based on the $k=10$ upper order statistic,
we assess the variability of the estimates $\widehat\Psi$
for different values of $k_n = \lfloor n ^ \nu \rfloor$,
with $\nu \in [0.2, 0.7]$.
Figure~\ref{fig:financial_robk} shows the estimated coefficients $\widehat \Psi$
for the pairs (EURCHF, NESN), (EURCHF, NOVN), and (EURCHF, ROGN),
with the corresponding
90\% bootstrap confidence intervals.
In the three plots, the black (blue) line corresponds to the estimated coefficient
$\widehat\Psi_{\text{EURCHF}, i}$ ($\widehat\Psi_{i, \text{EURCHF}}$),
with $i = $ NESN, NOVN, ROG.
We can interpret the difference between the black and blue lines
as a causal signal, since $\Psi_{ij} - \Psi_{ji} > 0$ if variable $i$ causes
variable $j$, for $i, j \in V$ (see Section~\ref{subsec:realvalbeta}).
For the pairs (EURCHF, NESN) and (EURCHF, NOVN) the blue and the black lines overlap
for all values of the upper order statistic $k$, and therefore any possible causal
effect is not identified by the estimated coefficient $\widehat\Psi$.
This result agrees with Example~\ref{ex_2} of
Section~\ref{subsec:diff_tails} which shows that the causal tail coefficients
do not identify a causal signal when the ancestor has a heavier tail than its
descendant --- as is the case for the pairs (EURCHF, NESN) and (EURCHF, NOVN).
In contrast, if we consider the pair (EURCHF, ROGN) we notice that the difference
$\widehat\Psi_{\text{EURCHF}, \text{ROGN}} - \widehat\Psi_{\text{ROGN}, \text{EURCHF}}$
is positive for all fractional exponents $\nu \leq 0.4$. This can be explained
by the fact that EURCHF and ROGN have comparable tail indices, and therefore it is easier for 
the coefficient $\Psi$ to detect a possible causal effect between the variables.
\ifjournal{In Section~\ref{app:financial_dynamics} of	the Supplementary Material,
}\else{In Appendix~\ref{app:financial_dynamics}, }\fi we show the dynamic evolution of the $\widehat\Psi$ coefficient across time. 

Given the highly complex nature of financial markets, we do not take the conclusion of this experiment as a definite answer but rather consider it as an indication for a possible causal relationship in this data.

\subsection{River data}\label{sec:river}
We apply the EASE algorithm to the average daily discharges of the rivers located in the upper Danube basin. This dataset has been studied in~\citet{asadi2015extremes}, \citet{eng2018a} and \citet{mhalla2019causal}, and it is made available by the Bavarian Environmental Agency (\url{http://www.gkd.bayern.de}).
We consider average daily discharges for 12 stations along the basin, representing the different tributaries and different sections of the Danube, and 11 of them are a subset of the 31 stations selected by~\citet{asadi2015extremes}.
We exclude some of the 31 stations that are spatially very close since those are highly dependent and almost indistinguishable.
For convenience, we name the stations with the same numbers used in~\citet{asadi2015extremes}.
While~\citet{asadi2015extremes} decluster the data prior to their analysis in order to obtain independent samples, we use all observations despite the possible temporal dependence.
In fact, for extreme value copulas, \cite{zou2019} show that the use of a larger but possibly dependent dataset can decrease the asymptotic estimation error.
Moreover, \citet{fawcett2007improved} argue that considering all exceedances over a high threshold reduces the bias of the maximum likelihood estimators compared to a declustered analysis. To account for the time dependence of the exceedances, they adjust the standard errors using methods presented by \citet{smith1990regional}. In this experiment, we compute the standard errors according to the adjustment proposed by~\citet{fawcett2007improved}.

The dataset spans from 1960 to 2009, where we consider only the summer months, i.e., June, July, and August. The rationale is that most of the extreme observations occur in summer due to heavy rainfall. The final dataset contains $n = 4600$ observations.
The rivers have an average volume that ranges between 20~$m^3/s$ (for the upstream rivers) and 1400~$m^3/s$ (for the downstream rivers). A map of the basin can be seen in Figure~\ref{fig:river_map} 
\ifjournal{of the Supplementary Material.
}\else{in Appendix~\ref{app:additional_figs}.}\fi
In order to implement our method, we first assess whether the equal tail index
assumption is satisfied.
To do so, we consider a regional model similar to the one
presented by~\citet{asadi2015extremes}. We split the stations into four separate
regions. Region 1 contains three stations in the southwest of the upper Danube basin
and the catchment areas are located at mid-altitude;
region 2 includes three stations in the Inn-Salzach basin whose tributaries are located
in high-altitude alpine regions; region 3 contains four stations along the main Danube
with large average water volume; region 4 comprises two stations in the north of the
Danube. For each region, we fit a Poisson point process likelihood~\citep[Chap.\ 7]{coles2001introduction} by considering exceedances over the 90\% quantile and
by constraining the shape parameter $\xi$ to be equal across the stations within the
same region.
To address the presence of temporal dependence in the exceedances, 
we adjust the standard errors as shown by~\citet{fawcett2007improved} and, based on these, we compute approximate confidence intervals.
For each region, the estimated shape parameter and the corresponding confidence
intervals are 0.167 (0.062, 0.273), 0.145 (0.047, 0.242), 0.133 (0.027, 0.239) and 0.229 (0.099, 0.358), respectively.
The fact that the confidence intervals overlap suggests that our assumption of
equal tail index across the variables is satisfied.
Moreover, all confidence intervals do not include the zero value, and therefore the data can be deemed to be heavy-tailed.

We perform two separate analyses to identify causal structures both in space and time.
Regarding the spatial structure, the goal is to recover the causal order of the network flow of the 12 stations on the rivers.
This is a non-trivial task for two reasons.
First, the water discharges at the stations can be confounded by rainfall that spreads out across the region of interest. Second, the large difference in water volume between the stations can further mask the causal structure of the water flow.
The true DAG of the spatial disposition of the stations is shown in Figure~\ref{fig:river_dag}. We run the EASE algorithm based on the $\G$ coefficient, considering the upper tails, and setting the number of exceedances to $k = 29 = \lfloor n ^ {0.4} \rfloor$. The estimated causal order $\widehat \pi^{-1} = $ (23, 32, 26, 28, 19, 21, 11, 9, 7, 14, 13, 1) is correct, and the corresponding fully connected DAG has an SID equal to 0. We also run ICA-LiNGAM and Pairwise LiNGAM on the same dataset, and we obtain an SID of 0 and 0.053, respectively. Clearly, in this example, the causal structure in the bulk of the distribution is the same as in the extremes.
\begin{figure}[h]
\centering
\includegraphics[scale=1]{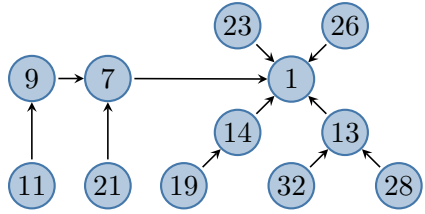}
\caption{DAG representing the spatial configuration of the stations across the upper Danube basin.}
\label{fig:river_dag}
\end{figure}
To assess the variability of our results, we compute the average SID of EASE over 50 bootstrap samples,
for different values of the threshold parameter $k = \lfloor n ^ \nu\rfloor$,
$\nu\in[0.2, 0.7]$ --- see Figure~\ref{fig:river_robk}.
From this figure, we observe that the fractional exponent $\nu \approx 0.4$
yields a good performance both in terms of SID and variability.
The value for the optimal fractional exponent agrees with the empirical
findings of Section~\ref{sec:simulation_study}.

Concerning the time-series analysis, we consider each station individually and try to recover
the direction of time from the lagged data. Consider an autoregressive (AR) process of order $p \geq 1$,
$$X_t = \sum_{j = 1}^{p} \beta_{j} X_{t - j} + \eps_j, \qquad t \geq 0,$$
where the $\eps_j$ are regularly varying with comparable tails. For a detailed discussion of such time series models we refer to~\citet{basrak2009regularly} and~\citet[Chapter~7]{embrechts2013modelling}.
\citet{Peters2009icml} prove that an AR($p$) process is time reversible, i.e., can be represented by an AR($p$) process in the reversed time direction,
if and only if the noise is Gaussian.
This means that for heavy-tailed random variables, one can in principle detect
the direction of time from the data.
For each station, we construct a dataset $D$ where the rows correspond to different
days and the columns to different lags. We denote by $X_0, X_1, \dots, X_6$
the columns containing
the current and lagged values of the station discharge.
We then run EASE on the dataset $D$ and recover a causal order $\widehat\pi$
over the seven variables $X_0, \dots, X_6$.
We say that the direction of time is correctly inferred if the estimated
causal order places the lags in the correct position, i.e.,
$\widehat\pi(j) < \widehat\pi(i)$ if $j < i$, for $i, j = 0, \dots, 6$.
EASE successfully recovers the direction of time for 11 out of the 12 stations.
As a comparison, ICA-LiNGAM, and Pairwise LiNGAM find the correct order for 9
and 11 stations, respectively.

So far, we have not considered a multivariate time series analysis of the dataset. 
On the one hand, in this particular application of the river data,
the effects among the stations are almost instantaneous --- due
to the closeness of the water catchment areas, the fact that we have daily values and the river speed, which is about ten kilometres per hour.
On the other hand, time information usually helps in estimating causal relations. For this 
reason, we apply multivariate Granger causality~\citep{granger1969investigating} to the river data, considering one day lag. 
For each pair of stations $(i, j)$, we say that $i$ Granger-causes $j$ if the 
corresponding $p$-value is significant at a 0.05 level, after the Bonferroni 
correction.
We sort the significant $p$-values in ascending order,
and we sequentially add the directed edge $(i, j)$ if nodes $i$ and $j$ are not
connected.
We continue until the skeleton of the resulting graph is connected, or all the $p$-values have been selected.
The resulting directed tree achieves an SID of 0.083.
Alternatively, if we sequentially add directed edges $(i, j)$ that do not create cycles (until all $p$-values have been selected) we obtain a DAG with an SID of 0.196.

\section{Discussion and future work}

In several real-world phenomena, the causal mechanisms between the variables appear more clearly during extreme events. Moreover, there are situations where the causal relationship in the bulk of the distribution differs from the structure in the tails.
We have introduced an algorithm, named extremal ancestral search (EASE), that is shown to consistently recover the causal order of a DAG from extreme observations only. EASE has the advantage of relying on the pairwise causal tail coefficient, and therefore it is computationally efficient. In addition, our algorithm can deal with the presence of hidden confounders and performs well
for small sample sizes and high dimensions.
The EASE algorithm is robust to model misspecifications,
such as nonlinear relationships in the bulk of the distribution,
and strictly monotone increasing transformations applied marginally to each variable.

This work sheds light on a connection between causality and extremes, and thereby opens new directions of research.
In particular, it might be interesting to study the properties of the causal tail coefficient
under
more general conditions.
This includes high-dimensional settings where the dimension grows with the sample size,
more general SCMs where the functional relations between the variables can be nonlinear,
and settings where the noise variables have lighter tails. For example, in future research,
one could combine EASE with regression techniques to obtain the complete DAG structure, 
compute the residuals and then test them for the assumption of common tail indices.

Another possible extension may consider multivariate time series data,
where the temporal order of cause and effect could help to estimate causal
relationships among variables, see, e.g.,~\citet{granger1969investigating}.
Future work might study how to combine our approach with the
Granger causality framework. For instance, one could first perform a
Granger causality analysis,
and then, apply EASE to the residuals of the vector autoregression model.
For a careful study in this direction,
it would also be necessary to investigate
the statistical properties of the residuals.

\section*{Acknowledgements}
We thank Cesare Miglioli, Stanislav Volgushev and Linbo Wang for helpful discussions. We are grateful to the
editorial team and two anonymous referees for constructive comments that improved the paper.
JP was supported by research grants from VILLUM FONDEN and the Carlsberg Foundation,
and SE was supported by the Swiss National Science Foundation.

\ifjournal{
  \begin{supplement}
  \stitle{Supplement to ``Causal Discovery in Heavy-Tailed Models''}
  \slink[doi]{??}
  \sdatatype{.pdf}
  \sfilename{??.pdf}
  \sdescription{
  The Supplementary Material~\citep{gnesupp2019} consists of six sections.
  Section~\ref{app:regvar}	summarises important facts about regularly varying random
  variables.
  Section~\ref{app:proofs} contains the proofs of the results of the paper.
  Section~\ref{app:minimax_search} illustrates how the EASE algorithm retrieves a causal order of
  a DAG.
  Section~\ref{app:simulation_experiments} describes the settings used in the simulation study. 
  Section~\ref{app:additional_figs2} contains additional figures and tables.
  Section~\ref{app:financial_dynamics} presents further results for Section~\ref{sec:financial}.
  }
  \end{supplement}
}\fi

\begin{appendix}

\ifjournal{
	\section{Figures}\label{app:additional_figs}
	
\begin{figure}[h!]
\centering
\includegraphics[scale=0.45]{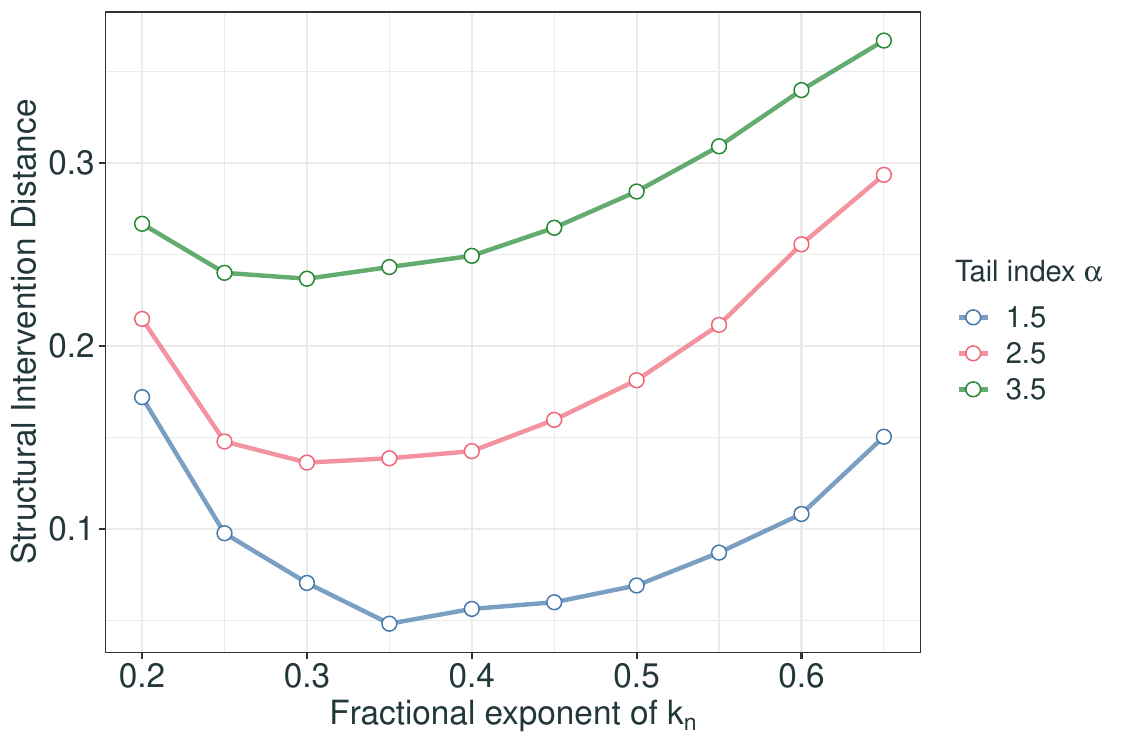}
\caption{The figure refers to Section~\ref{subsec:sim_results}. It shows the structural intervention distance (SID) of the EASE algorithm for different fractional exponents $\nu \in [0.2, 0.7]$ of $k_n = \lfloor n^\nu \rfloor$ and different tail indices $\alpha \in\{1.5, 2.5, 3.5\}$. Each point represents the SID measure averaged over 10 random samples for different sample sizes $n \in \{500, 1000, 10000\}$ and dimensions $p \in \{4, 7, 10, 15, 20, 30, 50\}$ in a linear SCM. 
In the experiments of Section~\ref{sec:simulation_study}, we set $k_n = \lfloor n^{0.4}\rfloor$.}
\label{fig:k-robustness}
\end{figure}

\begin{figure}[h!]
\centering
\includegraphics[scale=.45]{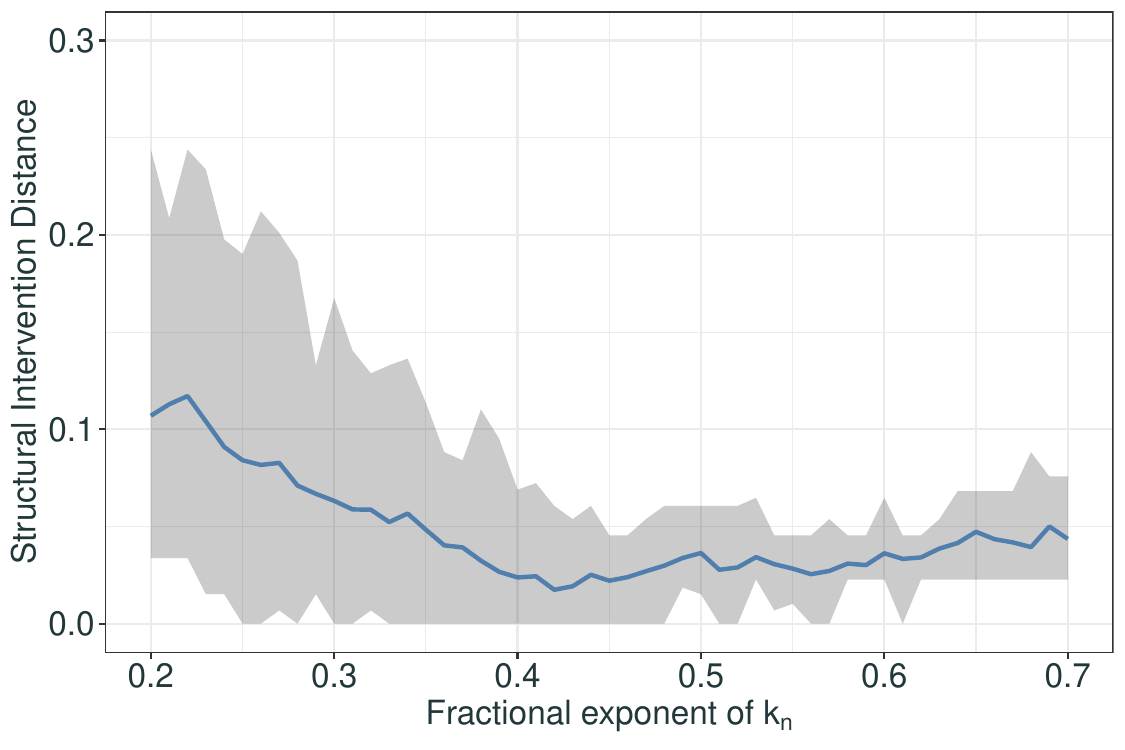}
\caption{The figure refers to Section~\ref{sec:river}. It shows the robustness of the structural intervention distance (SID) for varying fractional exponents $\nu \in [0.2, 0.7]$ of $k_n = \lfloor n ^ \nu\rfloor$.
Each point represents an average over 50 SID evaluations for the EASE algorithm, after bootstrapping the original dataset. The shaded interval corresponds to the 90\% bootstrap confidence interval.}
\label{fig:river_robk}
\end{figure}

\begin{figure}[h!]
\centering
\includegraphics[scale=0.6]{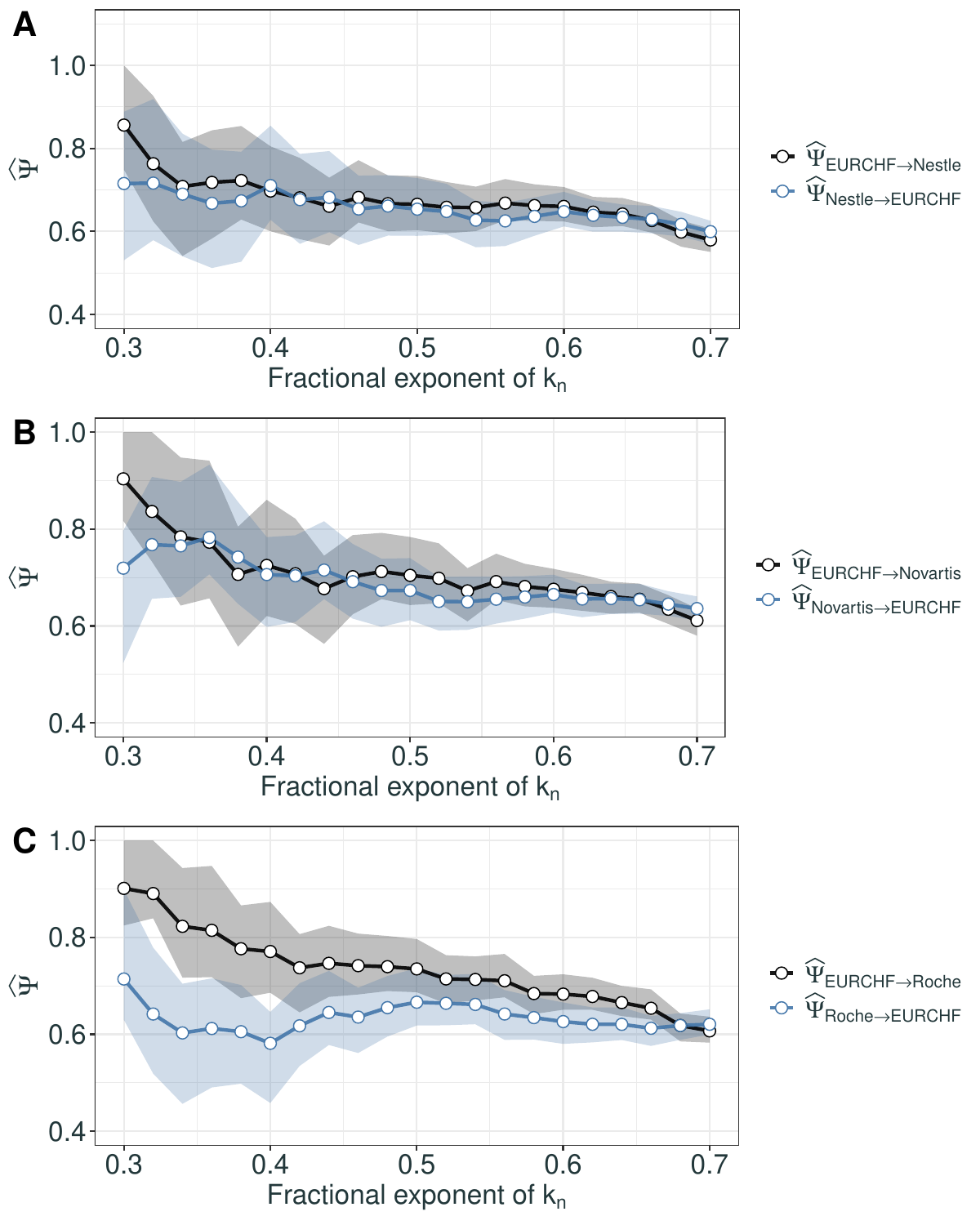}
\caption{The figure refers to Section~\ref{sec:financial}. It shows the variability of the estimated coefficients $\widehat\Psi$ for different fractional exponents $\nu \in [0.2, 0.7]$ of $k_n = \lfloor n ^ \nu \rfloor$. Each point represents the estimates of $\widehat\Psi$
based on the full dataset.
The shaded intervals correspond to the 90\% bootstrap confidence intervals
over 1000 repetitions.}
\label{fig:financial_robk}
\end{figure}

}\else{

\section{Some facts about regular variation}\label{app:regvar}
In the sequel, for any two functions $f$, $g: \mathbb{R} \rightarrow \mathbb{R}$,
we write $f \sim g$ if $f(x)/g(x) \to 1$ as $x\to\infty$. Also, we write $S_p := Y_1 + \dots + Y_p$, and $M_p := \max(Y_1, \dots, Y_p)$.

Consider independent random variables $Y_1,\dots, Y_p$ and assume that they have comparable upper tails, i.e., there exist  $c_j, \alpha > 0$ and $\ell \in \RV_0$ such that for all $j \in \{1,\dots, p\}$
  \begin{align}\label{eq:comp_tail1}
  \P(Y_j > x) \sim c_j \ell(x) x^{-\alpha},\quad x\to\infty.
  \end{align}

\begin{lemma}\label{lemma:tailconvp}
Let $Y_1, \dots, Y_p$ be real-valued independent regularly varying random variables with comparable tails. Then,
  $$\P(S_p > x) \sim \sum_{h = 1}^{p} \P(Y_h > x), \quad x\to\infty.$$
\end{lemma}
The proof
for $p=2$
of Lemma~\ref{lemma:tailconvp} can be found in \citet[][p.~278]{feller1971} and can be extended
to a general $p$
using induction.

An important property of regularly varying random variable is the max-sum-equivalence presented in the following lemma \citep[][Sec.\ 1.3.1]{embrechts2013modelling}.
\begin{lemma}\label{prop:maxsum}
Let $Y_1, \dots, Y_p$ be real-valued, independent regularly varying random variables with comparable tails. Then, as $x\to\infty$,
    \begin{equation*}
    \P(M_p > x) \sim \P(S_p > x).
    \end{equation*}
\end{lemma}

\begin{proof}
We can write, as $x\to\infty$,
  \begin{align*}
  \begin{split}
  \P\{M_p > x\}
  = &\ 1 - \P\left\{M_p \leq x\right\} = 1 - P\left(Y_1 \leq x, \dots,  Y_p \leq x\right)\\
  = &\  \sum_{h = 1}^{p} \P(Y_h > x) - \sum_{1\leq h < h' \leq p} \P(Y_h > x, Y_{h'} > x) + \Delta(x),
  \end{split}
  \end{align*}
where $\Delta(x)$ contains terms of higher order interactions of the sets $\{Y_j > x\}$, $j = 1, \dots, p$.
Because of independence, the probability
  \begin{align*}
  \P(Y_j > x, Y_{j'} > x) = o\{\P(S_p > x)\}.
  \end{align*}
Similarly, this holds for the terms in $\Delta(x)$. Recalling, by Lemma~\ref{lemma:tailconvp}, that $\P(Y_1 > x) + \dots + \P(Y_p > x) \sim \P(S_p > x)$, the result follows.
\end{proof}

\begin{lemma}\label{lemma:regvar1}
Let $Y_1, \dots, Y_p$ be real-valued independent regularly varying 
random variables with comparable tails.
Then, for $j = 1, \dots, p$,
  $$\P(Y_j > x, S_p > x) \sim \P(Y_j > x), \quad x\to\infty.$$
\end{lemma}

\begin{proof}
For any $x > 0$ and $\delta \in (0, 1/(2p - 2))$, we have, for $j = 1, \dots, p$,
  \begin{align*}
  \Bigg\{Y_j > x + (p - 1)\delta x, \bigcap_{h\neq j} \{Y_h > -\delta x\}\Bigg\}
  \subset  \{Y_j > x, S_p > x\}\subset \{Y_j > x\}.
  \end{align*}
Considering the upper bound, it holds, $j = 1, \dots, p$,
  \begin{align*}
  \P(Y_j > x, S_p > x) \leq \P(Y_j > x).
  \end{align*}
Regarding the lower bound, we get, as $x\to\infty$, $j = 1, \dots, p$,
  \begin{align*}
  \begin{split}
  \P(Y_j > x, S_p > x) 
  \geq &\ \P\left(Y_j > x + (p - 1)\delta x, 
  \bigcap_{h\neq j} \{Y_h > -\delta x\}\right)\\
  = &\ \P\left(Y_j > x + (p - 1)\delta x\right)
  \prod_{h\neq j}\P\left(Y_h > -\delta x\right).
  \end{split}
  \end{align*}
Dividing everything by $\P(Y_j > x)$, and letting first $x\to\infty$ and then
$\delta\downarrow 0$, we get the desired result.
\end{proof}

\begin{lemma}\label{lemma:regvar2}
Let $Y_1, \dots, Y_p$ be real-valued independent regularly varying random variables with comparable tails. Then, as $x\to\infty$,
  $$\P\{S_p > x, M_p \leq x\} = o\{\P(S_p > x)\}.$$
\end{lemma}

\begin{proof}
We first write $\P\{S_p > x, M_p \leq x\} = \P(S_p > x) - \P\{S_p > x, M_p > x\}.$
Let $I := \oneto{p}$. By definition of the maximum function 
and using the inclusion-exclusion principle, it follows that
  \begin{align*}
  \begin{split}
  \P\{S_p > x, M_p > x\}
  = &\ \P\left(S_p > x, \bigcup_{h\in I} \{Y_h > x\}\right)\\
  = &\ \sum_{h \in I} \P\left(S_p > x, Y_h > x\right)\\
   &\ - \sum_{1\leq h<h'\leq p}  \P\left(S_p > x, Y_h > x, Y_{h'} > x\right) + \Delta(x).
  \end{split}
  \end{align*}
Regarding the summands in first term, it holds by Lemma~\ref{lemma:regvar1}, $h\in I$,
  \begin{align*}
  \P\left(S_p > x, Y_h > x\right) \sim \P(Y_h > x),\quad x\to\infty.
  \end{align*}
The summands in the second term can be upper bounded by, $1 \leq h < h' \leq p$,
  \begin{align*}
  \P(Y_h > x, Y_{h'} > x) = o\{\P(S_p > x)\},\quad x\to\infty.
  \end{align*}
The same holds for $\Delta(x)$ which contains terms of higher order interactions of the sets $\{Y_j > x\}$, $j \in I$.
Putting everything together, we obtain
  \begin{align*}
  \begin{split}
  \P\{S_p > x, M_p \leq x\}
  \sim \P(S_p > x) - \sum_{h\in I} \P(Y_h > x) + o\{\P(S_p > x)\}
  = o\{\P(S_p > x)\},
  \end{split}
  \end{align*}
where in the last equality we used Lemma~\ref{lemma:tailconvp}.
\end{proof}

\section{Proofs}\label{app:proofs}
\subsection{Proof of Lemma~\ref{prop:p01}}\label{proof:p01}
\begin{proof}
Let $j, k \in V$ and $j\neq k$.
Recall that each variable $X_h$, $h\in V$, can be expressed as a weighted sum of the noise terms $\eps_1, \dots, \eps_p$ belonging to the ancestors of $X_h$, as shown in~\eqref{eq:scm-recursive}.
Therefore, we can write $X_j$ and $X_k$ as
  \begin{align*}
  \begin{split}
  X_j =  \sum_{h \in A_{jk}} \beta_{h\to j}\eps_h + \sum_{h \in A_{jk^{*}}} \beta_{h\to j}\eps_h,\\
  X_k = \sum_{h \in A_{jk}} \beta_{h\to k}\eps_h + \sum_{h \in A_{kj^{*}}} \beta_{h\to k}\eps_h,\\
  \end{split}
  \end{align*}
  where $A_{jk} = \An(j, G)\cap\An(k, G)$, $A_{jk^{*}} =  \An(j, G)\cap\An(k, G)^c$ and similarly for $A_{kj^{*}}$.
  We have
  \begin{align*}
  \E \left[F_k(X_k) \mathbf{1}\{ X_j > x \}\right] 
  = &\ \E \left[F_k(X_k) \mathbf{1}\left\{ X_j > x ,\bigcup_{h\in \An(j, G)} \left\{\beta_{h\to j}\eps_h > x\right\} \right\}\right] \\
  &\ +  \E \left[F_k(X_k) \mathbf{1}\left\{ X_j > x , \max_{h\in \An(j, G)} \left\{\beta_{h\to j}\eps_h\right\} \leq x\right\}\right].
  \end{align*}
The second summand can be bounded by
  \begin{align*}
  \P \left[  X_j > x , \max_{h\in \An(j, G)} \left\{\beta_{h\to j}\eps_h\right\} \leq x\right] = o\{\P(X_j > x)\},
  \end{align*}
by Lemma~\ref{lemma:regvar2}.
 For the first term, we use the inclusion-exclusion principle to write
 \begin{align*}
   \mathbf{1}\Bigg\{ X_j > x ,  &  \bigcup_{h\in \An(j, G)} \left\{\beta_{h\to j}\eps_h > x\right\} \Bigg\}
   = \sum_{h\in \An(j, G)} \mathbf{1}\left\{ X_j > x , \beta_{h\to j}\eps_h > x \right\} \\
   &\ - \sum_{h,h'\in \An(j, G), h<h'} \mathbf{1}\left\{ X_j > x , \beta_{h\to j}\eps_h > x, \beta_{h'\to j}\eps_{h'} > x \right\} + \Delta(x),
 \end{align*}
 where $\Delta(x)$ contains terms of higher order interactions of the sets $\{ \beta_{h\to j}\eps_h > x\}$, $h\in \An(j, G)$. The probability
 \begin{align*}
   &\ \P \left( X_j > x , \beta_{h\to j}\eps_h > x, \beta_{h'\to j}\eps_{h'} > x \right)\\
   \leq &\ \P\left(\beta_{h\to j}\eps_h > x, \beta_{h'\to j}\eps_{h'} > x \right)= o\{\P(X_j > x)\}.
 \end{align*}
The same holds for all finitely many terms in $\Delta(x)$. We further note
 that for all $h\in \An(j, G)$, by Lemma~\ref{lemma:regvar1},
  \begin{align*}
   \P(X_j > x , \beta_{h\to j}\eps_h > x) = \P(\beta_{h\to j}\eps_h > x) + o\{\P(X_j > x)\}.
 \end{align*}
 Putting everything together, we can rewrite
 \begin{align*}
   \E \left[F_k(X_k) \mathbf{1}\{ X_j > x \}\right] = & \sum_{h\in \An(j, G)} \E \left[F_k(X_k) \mathbf{1}\left\{ \beta_{h\to j}\eps_h > x \right\}\right] + o\{\P(X_j > x)\}\\
   = &\ \sum_{h \in A_{jk}} \E \left[F_k(X_k) \mathbf{1}\left\{ \beta_{h\to j}\eps_h > x \right\}\right]\\
   &\ +  \sum_{h \in A_{jk^{*}}} \E \left[F_k(X_k) \mathbf{1}\left\{ \beta_{h\to j}\eps_h > x \right\}\right] + o\{\P(X_j > x)\}.
 \end{align*}
 For $h \in A_{jk}$, let $c = \beta_{h\to j}/\beta_{h\to k}>0$, and note that
 for every $x > 0$,
  \begin{align*}
    \P\left( \beta_{h\to j}\eps_h > x \right) 
    \geq &\ \E \left[F_k(X_k) 
    \mathbf{1}\left\{\beta_{h\to j}\eps_h > x \right\}\right] \\
    \geq &\  \E \left[F_k(X_k) 
    \mathbf{1}\left\{ \beta_{h\to k}\eps_h > cx, X_k > cx\right\}\right] \\
    \geq &\ F_k(cx) \P \left(\beta_{h\to k}\eps_h > cx, X_k >  cx\right).
    \end{align*}
Therefore, using Lemma~\ref{lemma:regvar1} and that $F_k(cx) \to 1$ as $x\to \infty$, it follows that
  \begin{align*}
  \E \left[F_k(X_k) 
  \mathbf{1}\left\{\beta_{h\to j}\eps_h > x \right\}\right]
  \sim \P\left( \beta_{h\to j}\eps_h > x \right).
  \end{align*}
  On the other hand, for $h \in A_{jk^{*}}$, we have that $X_k$ and $\eps_h$ are independent,
  and therefore
  \begin{align*}
    \E \left[F_k(X_k) \mathbf{1}\left\{ \beta_{h\to j}\eps_h > x \right\}\right]
    = \frac12  \P\left( \beta_{h\to j}\eps_h > x \right),
    \quad x > 0.
  \end{align*}

  Consequently,
  \begin{align*}
    \G_{jk}
    = &\lim_{x\to \infty} \E \left[F_k(X_k) \mid  X_j > x \right]\\
    = &\  \lim_{x\to \infty} \sum_{h\in A_{jk}} \frac{\P\left(\beta_{h\to j}\eps_h > x \right)}{\P(X_j > x)} + \lim_{x\to \infty} \frac12 \sum_{h\in A_{jk^{*}}} \frac{\P\left(\beta_{h\to j}\eps_h > x \right)}{\P(X_j > x)}\\
    = &\ \frac{1}{2} + \frac{1}{2} \sum_{h\in A_{jk}}\lim_{x\to\infty}\frac{\P(\beta_{h\to j}\eps_h > x)}{\P(X_j > x)}
    = \frac{1}{2} + \frac{1}{2} \lim_{x \to \infty} \frac{\sum_{h\in A_{jk}} \beta_{h\to j}^{\alpha}\P(\eps_h > x)}{\sum_{h\in\An(j, G)} \beta_{h\to j}^{\alpha}\P(\eps_h > x)}\\
    = &\ \frac{1}{2} + \frac{1}{2} \frac{\sum_{h\in A_{jk}} \beta_{h\to j}^{\alpha}}{\sum_{h\in\An(j, G)} \beta_{h\to j}^{\alpha}},
 \end{align*}
where the second last equality follows from the fact that $\beta_{h\to j} \eps_h$, $h \in A_{jk}$, are independent regularly varying random variables,
see Lemma~\ref{lemma:tailconvp},
and
the last equality holds because we assume that the noise variables $\eps_j$, $j\in V$, have comparable tails; see Section~\ref{subsec:setup}.
\end{proof}

\subsection{Proof of Theorem~\ref{thm:6cases}}\label{proof:6cases}
\begin{proof}
Recall that
$\an(j, G) = \An(j, G)\setminus\{j\}$
and
define $A_{12} = \An(1, G) \cap \An(2, G)$.
\begin{description}
\item[$(a).$]
Suppose $X_1$ causes $X_2$, i.e., $1 \in \an(2, G)$. This implies that $\An(1, G) \subset \An(2, G)$ and thus $A_{12} = \An(1, G) \subset \An(2, G)$. By applying Lemma~\ref{prop:p01} we obtain $\G_{12} = 1$ and $\G_{21} \in (1/2, 1)$.

\setlength{\parindent}{3ex}
Conversely, suppose that $\G_{12} = 1$ and $\G_{21} \in (1/2, 1)$. If $\G_{12} = 1$ then the numerator and denominator of the second term in Lemma~\ref{prop:p01} must be equal and strictly positive. This implies that $A_{12} = \An(1, G) \cap \An(2, G) = \An(1, G) \neq \varnothing$. It follows that $\An(1, G) \subseteq \An(2, G)$, i.e., $1 \in \An(2, G)$. At the same time, if $\G_{21} \in (1/2, 1)$, then the numerator of the second term in Lemma~\ref{prop:p01} must be positive and smaller than the denominator. This means that $A_{12} \neq \varnothing$ and $A_{12} = \An(1, G) \cap \An(2, G) \subset \An(2, G)$. Thus, it follows that $\An(1, G) \subset \An(2, G)$. Therefore, $1 \in \an(2, G)$, that is, $X_1$ causes $X_2$.

\item[$(b).$] By symmetry, as case ($a$).

\item[$(c).$] Suppose there is no causal link between $X_1$ and $X_2$, i.e., $\An(1, G) \cap \An(2, G) = \varnothing$. Then, $A_{12} = \An(1, G) \cap \An(2, G) = \varnothing$ and by Lemma~\ref{prop:p01}, we obtain $\G_{12} = \G_{21} = 1/2$.

\setlength{\parindent}{3ex}
Suppose now that $\G_{12} = \G_{21} = 1/2$. This means that the numerator of the second term in Lemma~\ref{prop:p01} must be equal to zero. This implies that $A_{12} = \varnothing$ and therefore $\An(1, G) \cap \An(2, G) = \varnothing$, that is, there is no causal link between $X_1$ and $X_2$.

\item[$(d).$]
Suppose there is a node $j\notin\{1, 2\}$ such that $X_j$ is a common cause of $X_1$ and $X_2$, i.e., $j\in \an(1, G)$ and $j\in\an(2, G)$. Then $A_{12} = \An(1, G) \cap \An(2, G)$ is non-empty. Since $\An(1, G) \neq \An(2, G)$, it follows that $A_{12} \subset \An(i, G)$, for $i = 1, 2$. Thus, according to Lemma~\ref{prop:p01} we have $\G_{12}, \G_{21}\in (1/2, 1)$.

\setlength{\parindent}{3ex}
Conversely, suppose that $\G_{12}, \G_{21}\in (1/2, 1)$. If $\G_{12} \in (1/2, 1)$, then the numerator of the second term in Lemma~\ref{prop:p01} must be positive and smaller than the denominator. This implies that $A_{12} \neq \varnothing$ and $A_{12} = \An(1, G) \cap \An(2, G) \subset \An(1, G)$. Similarly, if $\G_{21} \in (1/2, 1)$, it follows that $A_{21} = A_{12} = \An(1, G) \cap \An(2, G) \subset \An(2, G)$. This implies that $\An(1, G) \neq \An(2, G)$ and they are not disjoint. Therefore, there exists a node $j\notin\{1, 2\}$ such that $j\in \an(1, G)$ and $j\in\an(2, G)$, i.e., $X_j$ is a common cause of $X_1$ and $X_2$.

\end{description}
\end{proof}

\subsection{Proof of Theorem~\ref{thm:Gamma_consistency}}\label{proof:Gamma_consistency}
\begin{proof}
  For simplicity we will write $k = k_n$ in the sequel. We only show the result for $\widehat \Gamma_{21}$, the proof for $\widehat \Gamma_{12}$ follows by symmetry. Recall that each variable $X_h$, $h\in V$, can be expressed as a weighted sum of the noise terms $\eps_1, \dots, \eps_p$ belonging to the ancestors of $X_h$, as shown in~\eqref{eq:scm-recursive}.
Therefore, we can write $X_1$ and $X_2$ as follows,
  \begin{align*}
  \begin{split}
  X_1 = &\ \sum_{h \in A} \beta_{h\to 1}\eps_h + \sum_{h \in A_{1}} \beta_{h\to 1}\eps_h,\\
  X_2 = &\ \sum_{h \in A} \beta_{h\to 2}\eps_h + \sum_{h \in A_{2}} \beta_{h\to 2}\eps_h,\\
  \end{split}
  \end{align*}
where $A = A_{12} = \An(1, G) \cap \An(2, G)$ and  $A_j = A_{jk^{*}} = \An(j, G) \setminus A$, for $j, k = 1, 2$. Thus, the estimator $\widehat \Gamma_{21}$ can be rewritten as
    \begin{align}\label{2terms}
    \begin{split}
    \widehat\G_{21}
    = &\ \frac{1}{k} \sum_{i = 1}^{n} \widehat F_1(X_{i1})\ \mathbf{1}\left\{X_{i2} > X_{(n-k), 2}, \bigcup_{h\in\An(2, G)} \left\{\beta_{h\to 2}\eps_{ih} > X_{(n-k), 2}\right\}\right\}\\
    &\ + \frac{1}{k} \sum_{i = 1}^{n} \widehat F_1(X_{i1})\ \mathbf{1}\left\{X_{i2} > X_{(n-k), 2}, \max_{h\in\An(2, G)} \beta_{h\to 2}\eps_{ih} \leq X_{(n-k), 2}\right\}\\
    = &\ S_{1, n} + S_{2, n}.
    \end{split}
    \end{align}
    Define the theoretical quantile function as
    \begin{align}\label{eq:U}
      U(x) = F^\leftarrow(1 - 1/x),\quad x > 1.
    \end{align}
    Recall that $X_{(n-k),2} = \widehat F_2^{\gets}(1 - k /n)$ is the $(n-k)$-th order statistic of $X_{12},\dots ,X_{n2}$, and as such an approximation to the theoretical quantile $U_2(n/k)$.
    Under the von Mises' condition \eqref{eq:cdfcondition} this convergence of the intermediate order statistics can be made rigorous \citep[][Theorem 2.2.1]{deh2006a}, namely
  \begin{align}\label{AN_intermediate}
    \sqrt{k} \left(\frac{X_{(n-k),2}}{ U_2\left(\frac{n}{k}\right)} - 1\right) \stackrel{d}{\longrightarrow} N(0,1/\alpha^2), \quad n\to\infty,
  \end{align}
  where $\alpha > 0$ is the tail index of the variables in the SCM. This implies in particular that $X_{(n-k),2} \to\infty$, $n \to \infty$. For any $\delta_1 > 0$ define the event
    \begin{align}\label{eq:setA}
           B_{n\delta_1} = \left\{\left| {X_{(n-k),2}}/{U_2\left(\frac{n}{k}\right)} -1 \right| < \delta_1 \right\},
    \end{align}
      and note that by \eqref{AN_intermediate} it holds that $\P(B_{n\delta_1}) \to 1$ as $n\to \infty$.

  Since the noise terms $\eps_h$ are independent regularly varying random variables with comparable tails, then, by Lemma~\ref{lemma:tailconvp}, we have
    \begin{align*}
    \P(X_2 > x) \sim \left(\sum_{h\in \An(2, G)}\beta_{h\to 2}^{\alpha}\right) \ell(x) x^{-\alpha}  =: c_2 \ell(x) x^{-\alpha}.
    \end{align*}
  Furthermore, from \citet[][Prop.~0.8]{res2008} it holds, for all $h\in \An(2, G)$ and $x > 0$,
    \begin{align}\label{U}
      \P\{\eps_h > x U_2(t)\} \sim x^{-\alpha}(c_2 t)^{-1}, \quad t\to \infty,
    \end{align}
    where $U_2$ is defined as in~\eqref{eq:U} for the distribution function $F_2$.

\label{pg:second-term}
  We treat the two terms in \eqref{2terms} separately. We can upper bound the absolute value of the second term, for any $\tau,\delta_1 > 0$, by
    \begin{align*}
    \begin{split}
    \P(|S_{2, n}|>\tau)
    \leq &\ \P(B_{n\delta_1}^{c})
    + \P\bigg(\frac{1}{k}\sum_{i = 1}^{n} \mathbf{1}\Big\{X_{i2} > U_2\left(\frac{n}{k}\right)(1 - \delta_1), \\
    &\ \max_{h\in\An(2, G)} \beta_{h\to 2} \eps_{ih} \leq U_2\left(\frac{n}{k}\right)(1 + \delta_1)\Big\} > \tau\bigg)\\
    \end{split}
    \end{align*}
   where $\P(B_{n\delta_1}^{c})\to 0$ as $n\to \infty$ by \eqref{AN_intermediate}. The limit superior of the second term, as $n\to\infty$, can be bounded with Markov's inequality by
    \begin{align*}
    \lim_{n\to \infty}
    &\ \frac{n}{\tau k} P\left\{X_2 > (1-\delta_1)U_2\left(\frac{n}{k}\right), \max_{h\in\An(2, G)} \beta_{h\to 2} \eps_{h} \leq  (1+\delta_1)U_2\left(\frac{n}{k}\right) \right\} \\
    \leq &\ \lim_{n\to \infty}  \frac{1}{\tau} \left(\P\left\{X_2 > U_2\left(\frac{n}{k}\right)\right\}\right)^{-1} \Bigg[\P\left\{X_2 > (1-\delta_1)U_2\left(\frac{n}{k}\right)  \right\}\\
    &\ - \P\left\{X_2 > (1 + \delta_1)U_2\left(\frac{n}{k}\right), \max_{h\in\An(2, G)} \beta_{h\to 2} \eps_{h} >  (1+\delta_1)U_2\left(\frac{n}{k}\right) \right\} \Bigg]\\
    = &\ \frac{ (1-\delta_1)^{-\alpha} - (1+\delta_1)^{-\alpha}}{\tau}.
    \end{align*}
  The last equality holds because $X_2$ is regularly varying with index $\alpha$ and because, by Lemma~\ref{lemma:regvar2}, we have
    \begin{align*}
    \P\left(X_2 > x, \max_{h\in\An(2, G)} \beta_{h\to 2}\eps_h > x\right) \sim \P(X_2 > x),\quad x\to\infty.
    \end{align*}
  Since $\delta_1, \tau > 0$ are arbitrary, it follows that $S_{2, n} = o_P(1)$.

  For the first term, we use the inclusion-exclusion principle to write
    \begin{align*}
    S_{1, n}
    = &\ \sum_{h\in \An(2, G)} \frac{1}{k}\sum_{i=1}^{n} \widehat F_1(X_{i1}) \mathbf{1}\left\{ X_{i2} > X_{(n-k), 2} , \beta_{h\to 2}\eps_{ih} > X_{(n-k), 2} \right\}\\
    &\ - \sum_{h, h' \in\An(2, G),h<h'} \frac{1}{k}\sum_{i=1}^{n} \widehat F_1(X_{i1}) \mathbf{1}\big\{ X_{i2} > X_{(n-k), 2} , \beta_{h\to 2}\eps_{ih} > X_{(n-k), 2},\\
    &\ \beta_{h'\to 2}\eps_{ih'} > X_{(n-k), 2} \big\} + \frac{1}{k}\sum_{i=1}^{n} \widehat F_1(X_{i1}) \Delta_{i}\left(X_{(n-k), 2}\right)\\
    = &\ T_{1, n} + T_{2, n} + T_{3, n},
    \end{align*}
    where $\Delta_{i}(X_{(n-k), 2})$ contains terms of higher order interactions of the sets $\{ \beta_{h\to 2}\eps_{ih} > X_{(n-k), 2}\}$, $h\in \An(2, G)$, $i = 1, \dots, n$. First, we show that the terms $T_{2, n}$ and $T_{3, n}$ are $o_P(1)$. Considering $T_{2, n}$, for each $h, h' \in\An(2, G)$, $h < h'$, define
    \begin{align*}
    T_{2, n}^{(h, h')}
    = &\ \frac{1}{k}\sum_{i=1}^{n} \widehat F_1(X_{i1}) \mathbf{1}\big\{ X_{i2} > X_{(n-k), 2} , \beta_{h\to 2}\eps_{ih} > X_{(n-k), 2},\\
    &\ \beta_{h'\to 2}\eps_{ih'} > X_{(n-k), 2} \big\}.
    \end{align*}
  We can upper bound its absolute value, for any $\tau,\delta_1 > 0$, by
    \begin{align*}
    \P(&|T_{2, n}^{(h, h')}| > \tau)
    \leq \P(B_{n\delta_1}^{c})\\
    &\ + \P\bigg(\frac{1}{k}\sum_{i=1}^{n} \mathbf{1}\big\{\beta_{h\to 2}\eps_{ih} > (1-\delta_1)U_2\left(\frac{n}{k}\right), \beta_{h'\to 2}\eps_{ih'} > (1-\delta_1)U_2\left(\frac{n}{k}\right) \big\}\bigg).
    \end{align*}
  By Markov's inequality,
    \begin{align*}
    \limsup_{n\to\infty} &\ \P(|T_{2, n}^{(h, h')}| > \tau)
    \leq \lim_{n\to\infty} \P(B_{n\delta_1}^{c})\\
    &\ + \lim_{n\to\infty}\frac{n}{\tau k} \P\bigg\{\beta_{h \to 2}\eps_h > (1-\delta_1)U_2\left(\frac{n}{k}\right)\bigg\} \P\bigg\{\beta_{h' \to 2}\eps_{h'} > (1-\delta_1)U_2\left(\frac{n}{k}\right)\bigg\}\\
    = &\ \lim_{n\to\infty} \frac{k(1 - \delta_1)^{-2\alpha}p_{2h}p_{2h'}}{n\tau} = 0,
    \end{align*}
  where in the last line we used property~\eqref{U} and the fact
  that $k / n \to 0$ 
  as $n\to\infty$, and $p_{h2}, p_{2h'}$ are defined as
    \begin{align}\label{eq:constants}
    \quad p_{2h} =\frac{\beta_{h\to 2}^{\alpha}}{c_2}, \quad h\in\An(2, G).
    \end{align}
  Since $\delta_1, \tau$ are arbitrary, putting together the finitely many terms $h < h'$ where $h, h'\in\An(2, G)$, it follows that $T_{2, n} \stackrel{P}{\longrightarrow} 0$. Using a similar argument as the one for $T_{2, n}$, one can show that $T_{3, n} \stackrel{P}{\longrightarrow} 0$.

  We want to show that $T_{1, n} \stackrel{P}{\longrightarrow} \G_{21}$. Rewrite
    \begin{align*}
    T_{1, n}
    = &\ \sum_{h\in A} \frac{1}{k}\sum_{i=1}^{n} \widehat F_1(X_{i1}) \mathbf{1}\left\{\beta_{h\to 2}\eps_{ih} > X_{(n-k), 2} \right\}\\
    &\ + \sum_{h\in A_2} \frac{1}{k}\sum_{i=1}^{n} \widehat F_1(X_{i1}) \mathbf{1}\left\{\beta_{h\to 2}\eps_{ih} > X_{(n-k), 2} \right\}\\
    &\ - \sum_{h\in \An(2, G)} \frac{1}{k}\sum_{i=1}^{n} \widehat F_1(X_{i1}) \mathbf{1}\left\{\beta_{h\to 2}\eps_{ih} > X_{(n-k), 2},  X_{i2} \leq X_{(n-k), 2} \right\}\\
    = &\ U_{1, n} + U_{2, n} + U_{3, n}.
    \end{align*}
    Using an argument similar to the one for $S_{2, n}$ on page~\pageref{pg:second-term}, 
    one can show that $U_{3, n} \stackrel{P}{\longrightarrow} 0$.

  Regarding $U_{1, n}$, for each $h\in A$, define
    \begin{align*}
    U_{1, n}^{(h)} = \frac{1}{k}\sum_{i=1}^{n} \widehat F_1(X_{i1}) \mathbf{1}\left\{\beta_{h\to 2}\eps_{ih} > X_{(n-k), 2} \right\}.
    \end{align*}
  Note that, for $h\in A$, we have that both $\beta_{h \to 2} > 0$ and $\beta_{h\to 1} > 0$, therefore we can bound
    \begin{align}\label{eq:Un-bound}
    \widehat F_1\left(cX_{(n-k), 2}\right) V_n^{(h)} - W_n^{(h)}\leq U_{1, n}^{(h)} \leq V_n^{(h)},
    \end{align}
  where $c = \beta_{h\to 1}/\beta_{h\to 2} > 0$ and
    \begin{align}\label{eq:Un-bound-def}
    \begin{split}
    V_n = &\ V_n^{(h)} = \frac{n}{k}\left[1 - \widehat F_{\eps_h}\left(\frac{X_{(n-k), 2}}{\beta_{h\to 2}}\right)\right],\\
    W_n = &\ W_n^{(h)} = \frac{1}{k}\sum_{i=1}^{n} \mathbf{1}\left\{\beta_{h\to 1}\eps_{ih} > cX_{(n-k), 2}, X_{i1}\leq cX_{(n-k), 2}\right\}.
    \end{split}
    \end{align}
    We show first that $V_n$ converges in probability to $p_{2h}$, defined in~\eqref{eq:constants}. This is motivated by the fact that $V_n$ is the empirical version of
    \begin{align*}
    t\P\left\{\beta_{h\to 2}\eps_h >  U_2(t)\right\} \to \beta_{h\to 2}/c_2 = p_{2h},\quad t\to\infty,
    \end{align*}
  where the limit follows from property~\eqref{U}. We will study the asymptotic properties of $V_n$. For $x > 0$ define
    \begin{align*}
    v_n(x) = \frac{n}{k}\left[1 - \widehat F_{\eps_h}\left(\frac{x U_2\left(\frac{n}{k}\right)}{\beta_{h\to 2}}\right)\right] = \frac{1}{k}\sum_{i=1}^{n}\mathbf{1}\left\{\beta_{h\to 2}\eps_{ih} > x U_2\left(\frac{n}{k}\right)\right\},
    \end{align*}
  which is a nonincreasing function of $x$. Observe that on the set $B_{n\delta_1}$ defined in~\eqref{eq:setA}, we may bound the random variable $V_n$ by
    \begin{align}\label{eq:Vn-bound}
    v_n(1 + \delta_1) \leq V_n = v_n\left(\frac{X_{(n-k), 2}}{U_2\left(\frac{n}{k}\right)}\right) \leq v_n(1 - \delta_1).
    \end{align}
  We further compute the limits as $n\to \infty$
    \begin{align*}
    \E\{v_n(x)\} &\ = \frac{n}{k} \P\left[\beta_{h\to 2}\eps_h > x U_2\left(\frac{n}{k}\right)\right] \to p_{2h} x^{-\alpha},\\
    \Var \left\{v_n(x)\right\} &=   \frac{n}{k^2} \P\left[\beta_{h\to 2}\eps_h > x U_2\left(\frac{n}{k}\right)\right] \P\left[\beta_{h\to 2}\eps_h \leq x U_2\left(\frac{n}{k}\right)\right] = O(1/k) \to 0.
    \end{align*}
  An application of Chebyshev's inequality yields
    \begin{align}\label{eq:vn}
    v_n(1+\delta_1)  \stackrel{P}{\longrightarrow} p_{2h} (1+\delta_1)^{-\alpha}, \quad v_n(1-\delta_1) \stackrel{P}{\longrightarrow} p_{2h} (1-\delta_1)^{-\alpha}, \quad n \to \infty.
    \end{align}
  For some $\tau > 0$, choose $\delta_1>0$ such that $p_{2h} (1+\delta_1)^{-\alpha} > p_{2h} - \tau$ and $p_{2h} (1-\delta_1)^{-\alpha} < p_{2h} + \tau$. Then with \eqref{eq:Vn-bound}, \eqref{eq:vn} and the fact that $\P(B_{n\delta_1}^c) \to 0$ for $B_{n\delta_1}$ in \eqref{eq:setA}, we conclude
    \begin{align}\label{eq:Vn}
    \begin{split}
    \P(| V_n - p_{2h}| > \tau)
    \leq &\ \P(B_{n\delta_1}^c) + \P\{ v_n(1+\delta_1)  < p_{2h} - \tau\}\\
    &\ + \P\{ v_n(1-\delta_1)  > p_{2h} + \tau\} \to 0, \quad n\to\infty,
    \end{split}
    \end{align}
  that is, $V_n$ converges in probability to $p_{2h}$, as $n\to\infty$.

  Furthermore, using~\eqref{eq:Un-bound} we will now show
    \begin{align}\label{eq:Fn}
    \widehat F_1\left(cX_{(n-k), 2}\right) \stackrel{P}{\longrightarrow} 1.
    \end{align}
Indeed,
  for any $\tau, \delta_1 > 0$, as $n\to \infty$
    \begin{align*}
    \P\left\{\left|\widehat F_1\left(cX_{(n-k), 2}\right) - 1\right| > \tau \right\}
    \leq &\  \P\left(B_{n\delta_1}^{c}\right)
    + \P\left[\widehat F_1\left\{c(1 + \delta_1)U_2\left(\frac{n}{k}\right)\right\} > 1 + \tau \right]\\
    &\ +  \P\left[\widehat F_1\left\{c(1 - \delta_1)U_2\left(\frac{n}{k}\right)\right\} < 1 - \tau \right]\to 0,
    \end{align*}
  since $\widehat F_1(x)$ converges in probability to $F_1(x)$ for all $x \in \mathbb R$, and $U_2(n/k)\to\infty$, as $n\to\infty$.
  Moreover, with a similar argument as for $S_{2, n}$, one can show that $W_{n}$ defined in~\eqref{eq:Un-bound-def} converges in probability to 0, as $n\to\infty$.

  Putting everything together, using~\eqref{eq:Un-bound}, \eqref{eq:Vn} and \eqref{eq:Fn} we conclude that $U_{1, n}^{(h)} \stackrel{P}{\longrightarrow} p_{2h}$, $h\in A$, and thus
    \begin{align}\label{eq:U1n}
    U_{1, n} \stackrel{P}{\longrightarrow} \sum_{h\in A} p_{2h} = \frac{\sum_{h\in A} \beta_{h\to 2}^{\alpha}}{\sum_{h\in\An(2, G)} \beta_{h\to 2}^{\alpha}}.
    \end{align}

  Considering the term $U_{2, n}$, for each $h\in A_2$, define
    \begin{align*}
    U_{2, n}^{(h)} = \frac{1}{k}\sum_{i = 1}^{n} \widehat F_1(X_{i1})\mathbf{1}\left\{\beta_{h\to 2} \eps_{ih} > X_{(n - k), 2}\right\}.
    \end{align*}
  On the event $B_{n\delta_1}$, we can bound $U_{2, n}^{(h)}$ by
    \begin{align*}
    u_{2, n}(1+\delta_1) \leq U_{2, n}^{(h)} \leq u_{2, n}(1-\delta_1),
    \end{align*}
  where we let, for all $x > 0$,
    \begin{align*}
     u_{2, n}(x) =  \frac{1}{k}  \sum_{i=1}^n  \widehat F_1(X_{i1}) \mathbf{1} \left\{\beta_{h\to 2} \eps_{ih} > x U_2\left(\frac{n}{k}\right)\right\}.
    \end{align*}
  Since $X_1$ is independent of $\eps_h$, for $h\in A_2$, the values in the sum can be seen as $M(x) = M_n(x)$ random samples out of $\{1/n,\dots, 1\}$ without replacement, where $M(x)$ is Binomial with success probability
    \begin{align*}
    \P\left\{ \beta_{h\to 2} \eps_{h} > x U_2\left(\frac{n}{k}\right)\right\} \sim \frac{k p_{2h} x^{-\alpha}}{n}, \quad n \to \infty.
    \end{align*}
    Let $Z_{ni}$ 
    be random samples out of $\{1/n,\dots, 1\}$ 
    without replacement, for all $i=1,\dots, n$, $n\in \mathbb N$. Then $u_{2, n}(x)$
    has the same distribution as
    \begin{align*}
      \frac{1}{k} \sum_{i=1}^{M(x)} Z_{ni}.
    \end{align*}
  By a similar argument as in \eqref{eq:vn} the distribution of $M(x)$ satisfies for any fixed $x\in (0,\infty)$
    \begin{align}\label{eq:binomasy}
    \frac{M(x)}{m(x)} \stackrel{P}{\longrightarrow} 1, \quad n\to\infty,
    \end{align}
  where $m(x) = m_n(x) = \lceil{k p_{2h} x^{-\alpha}}\rceil$.
  Thus, for any $\delta_2 > 0$ and any $x >0$, the probability of the event
    \begin{align*}
    C_{n\delta_2,x} = \left\{\left| \frac{M(x)}{m(x)} -1 \right| < \delta_2 \right\}
    \end{align*}
    converges to 1 as $n\to\infty$. 
    Consider the quantity
    \begin{align}\label{Vin}
      \widetilde u_{2, n}(x) = \frac{1}{k} \sum_{i=1}^{m(x)} Z_{ni}.
    \end{align}
   Theorem 5.1 in \cite{ros1965} states that the limit in probability of this sum of samples without replacement is the same as the corresponding sum of samples with replacement. Therefore, for any $x\in(0,\infty)$, we have the convergence in probability  
    \begin{align}\label{eq:asynorm}
    \widetilde u_{2, n}(x)  = \frac{m(x)}{k} \frac{1}{m(x)} \sum_{i=1}^{m(x)} Z_{ni} \stackrel{P}{\longrightarrow}  \frac{1}{2} p_{2h}x^{-\alpha}.
    \end{align}
  For $\tau > 0$, choose $\delta_1, \delta_2>0$ small enough such that
    \begin{align}\label{d1d2}
     p_{2h}(1-\delta_1)^{-\alpha}(1+\delta_2)/2 < p_{2h}/2 + \tau.
    \end{align}
  Then we can bound the probability
    \begin{align}\label{eq:boundS2}
    \begin{split}
    \limsup_{n\to\infty} &\ \P\left(U_{2, n}^{(h)} - p_{2h}/2  > \tau\right)\\
    \leq &\ \limsup_{n\to\infty}\P\left\{u_{2, n}(1 - \delta_1) > p_{2h}/2 + \tau\right\} + \P(B_{n\delta_1}^{c})\\
    \leq &\ \lim_{n\to\infty}\P\left[ \widetilde u_{2, n}\{(1-\delta_1)(1+\delta_2)^{-1/\alpha}\} > p_{2h}/2 + \tau\right]
     + \P(B_{n\delta_1}^{c}) + \P(C_{n\delta_2,1-\delta_1}^{c})\\
    = &\ 0,
    \end{split}
    \end{align}
  since $\P(B_{n\delta_1}^{c})$ and $\P(C_{n\delta_2, 1-\delta_1}^{c})$ converge to 0 as $n\to\infty$, and the last term converges to $0$ as a consequence of \eqref{eq:asynorm} and \eqref{d1d2}. Similarly, we can show that
    \begin{align*}
    \limsup_{n\to\infty} \P\left(U_{2, n}^{(h)} - p_{2h}/2  < -\tau \right) = 0,
    \end{align*}
  and since $\tau > 0$ is arbitrary, $U_{2, n}^{(h)}\stackrel{P}{\longrightarrow} p_{2h}/2$, for $h\in A_2$. Therefore,
    \begin{align*}
    U_{2, n} \stackrel{P}{\longrightarrow} \sum_{h\in A_2} \frac{p_{2h}}{2} = \frac{1}{2}\frac{\sum_{h\in A_2} \beta_{h\to 2}^{\alpha}}{\sum_{h\in\An(2, G)} \beta_{h\to 2}^{\alpha}},
    \end{align*}
  and $\widehat \G_{21}  \stackrel{P}{\longrightarrow} \G_{21}$.

\end{proof}

\subsection{Proof of Proposition~\ref{prop:minimax_oracle}}\label{proof:minimax_oracle}
\begin{proof}
Let $s \in S := \{1, \dots, p\}$ and denote by $i_{s} \in V$ the node chosen by the algorithm \emph{at} step $s$.
Denote by
  \begin{align*}
  H_s =
  \begin{cases}
  \varnothing,           &\ s = 1,\\
  \{i_1, \dots, i_{s-1}\}, &\ s > 1,\\
  \end{cases}
  \end{align*}
the set of nodes chosen by the algorithm \emph{before} step $s$.
Let $G_{s} = (V_{s}, E_{s})$ be the subgraph of $G$ obtained by removing the nodes
$H_s$
that are
already chosen, i.e., $V_{s} = V\setminus H_s$ and $E_{s} = E \cap (V_{s}\times V_{s})$. Furthermore, define the score minimized by the algorithm to choose the node at step $s$,
  \begin{align*}
  M_i^{(s)} = \max_{j\in V_{s}\setminus\{i\}}\G_{ji},\quad \forall i \in V_{s}.
  \end{align*}
We want to show that EASE is a procedure that, for all $s \in S$, satisfies the statement
  \begin{align}\label{eq:minimax_orac}
  \Xi(s) := \left(i_{s} \in \arg\min_{i \in V_s} M_i^{(s)} \implies \an(i_{s}, G) \subseteq H_s\right).
  \end{align}
We use strong induction. Namely, we prove that if $\Xi(s')$ holds for \emph{all} natural numbers $s' < s$, then $\Xi(s)$ holds, too.

Fix $s \in S$ and suppose that for all $s' \in S$, with $s' < s$, $\Xi(s')$ holds. Assume $\varphi:= i_{s} \in \arg\min_{i \in V_s} M_i^{(s)}$ and $\an(\varphi, G) \not\subseteq H_s$. Then, there exists a node $j \in V_s$ such that $j\in \an(\varphi, G)$, and by Theorem~\ref{thm:6cases}, $\G_{j\varphi} = 1$. It follows $M_{\varphi}^{(s)} = 1$.

Also, since $G_s$ is a DAG, there exists a node $\ell \in V_s$ such that $\an(\ell, G_s) = \varnothing$. If $\Xi(s')$ holds for every natural number $s' < s$, then $\an(\ell, G) \subseteq H_s$. Suppose not. Then, there exists a node $j\in V_s$ such that $j\in \an(\ell, G)$. Note that since $j\in \an(\ell, G)$ and $j\notin \an(\ell, G_s)$, there exists a directed path from $j$ to $\ell$ in $G$ that is absent in $G_s$. Thus, there exists a node $h \in H_s$ that lies on such path, and it follows that $j \in \an(h, G)$, which is a contradiction. Since $\an(\ell, G)\subseteq H_s$, by Theorem~\ref{thm:6cases} it holds $M_{\ell}^{(s)} < 1 = M_{\varphi}^{(s)}$, which is a contradiction.

Since $s\in S$ was arbitrary, we have proved that $\Xi(s)$ holds for all $s \in S$. Furthermore, note that $\Xi(1)$ holds as a special case of the argument above.
Hence, we conclude that, for all $s', s \in S$,
  \begin{align*}
  s' < s \implies i_{s} \notin \an(i_{s'}, G)
  \end{align*}
and therefore $\pi(i_{s}) = s$ is a causal order of $G$.
\end{proof}

\subsection{Proof of Proposition~\ref{prop:minimax_bound}}\label{proof:minimax_bound}
\begin{proof}
If $\widehat \pi\not\in\Pi_G$,
then there exists a node $i\in V$ that is chosen before one of
its ancestors $u\in\an(i, G)$, i.e., $\widehat\pi(i) < \widehat\pi(u)$.
Therefore, there exists a non-empty set $\tilde V \subseteq V$, with $i, u \in \tilde V$, such that
  \begin{align}\label{eq:ibest}
  i\in\arg\min_{i'\in \tilde V}\max_{u'\in\tilde V\setminus\{i'\}} \widehat\G_{u'i'}.
  \end{align}
Furthermore, since $G$ is a DAG, there exists a node $j\in \tilde V$ with no ancestors in $\tilde V$. Let $v \in \tilde V\setminus\{j\}$, and note that $v\notin\An(j, G)$. Thus, by~\eqref{eq:ibest}, it follows that $\widehat\G_{vj} - \widehat\G_{ui} \geq 0$.
Define $\widehat\Delta_{ij} := |\widehat\G_{ij} - \G_{ij}|$ and note that the event $\widehat\G_{vj} - \widehat\G_{ui} \geq 0$ can be bounded by
  \begin{align*}
  \begin{split}
  &\ \left\{\widehat\Delta_{vj}
  \geq \frac{\G_{ui} - \G_{vj}}{2} \right\}  
  \cup
  \left\{\widehat\Delta_{ui}
  \geq \frac{\G_{ui} - \G_{vj}}{2} \right\}
  \subseteq
  \left\{\widehat\Delta_{vj}
  \geq \frac{1 - \eta}{2} \right\}  
  \cup
  \left\{\widehat\Delta_{ui}
  \geq \frac{1 - \eta}{2} \right\},
  \end{split}
  \end{align*}
where $\G_{ui} = 1$ because $u\in\an(i, G)$, and 
$\G_{vj} \leq \eta < 1$.
Therefore,
  \begin{align*}
  \begin{split}
  \P(\widehat\pi\not\in\Pi_G)
  \leq &\ \sum_{j\in V, \ v \notin \An(j, G)}
  \P\left(\widehat\Delta_{vj} > \frac{1 - \eta}{2}\right)
  + \sum_{i\in V,\ u \in \an(i, G)}
  \P\left(\widehat\Delta_{ui} > \frac{1 - \eta}{2}\right)\\
  \leq &\ p^2 \max_{i, j\in V:i\neq j}\P\left(\widehat\Delta_{ij}
  > \frac{1 - \eta}{2}\right).
  \end{split}
  \end{align*}
This completes the proof of Proposition~\ref{prop:minimax_bound}.
\end{proof}

\subsection{Proof of Lemma~\ref{lemma:psi-coeff}}\label{proof:psi-coeff}
\begin{proof}
Recall that each variable $X_h$, $h\in V$, can be expressed as a weighted sum of the noise terms $\eps_1, \dots, \eps_p$ belonging to the ancestors of $X_h$, as shown in~\eqref{eq:scm-recursive}.
Therefore, we can write $X_j$ and $X_k$ as follows,
 \begin{align*}
  \begin{split}
  X_j =  \sum_{h \in A_{jk}} \beta_{h\to j}\eps_h + \sum_{h \in A_{jk^{*}}} \beta_{h\to j}\eps_h,\\
  X_k = \sum_{h \in A_{jk}} \beta_{h\to k}\eps_h + \sum_{h \in A_{kj^{*}}} \beta_{h\to k}\eps_h,\\
  \end{split}
  \end{align*}
where $A_{jk} = \An(j, G)\cap\An(k, G)$, $A_{jk^{*}} =  \An(j, G)\cap\An(k, G)^c$ and similarly for $A_{kj^{*}}$.

We treat the two terms of~\eqref{eq:psi-decomp} separately.
Consider the first term,
  \begin{align*}
  \Psi_{jk}^{+} = \lim_{x\to \infty}\frac{1}{2}\E\left[\sigma(F_k(X_k)) \mid X_j > x\right],\quad j, k\in V.
  \end{align*}
Since $X_j$, $j\in V$, are regularly varying with index $\alpha > 0$, using similar arguments as in Lemma~\ref{prop:p01}, we can write
  \begin{align}\label{eq:ext-1}
  \begin{split}
  \E [\sigma(F_k(X_k)) & \mathbf{1}\{ X_j > x \}]\\
  = &\ \sum_{h\in \An(j, G)} \E \left[\sigma(F_k(X_k)) \mathbf{1}\left\{ \beta_{h\to j}\eps_h > x \right\}\right] + o\{\P(X_j > x)\}\\
  = &\ \sum_{h \in A_{jk}} \E \left[\sigma(F_k(X_k)) \mathbf{1}\left\{ \beta_{h\to j}\eps_h > x \right\}\right]\\
  &\ + \sum_{h \in A_{jk^{*}}} \E \left[\sigma(F_k(X_k)) \mathbf{1}\left\{ \beta_{h\to j}\eps_h > x \right\}\right] + o\{\P(X_j > x)\}.
  \end{split}
  \end{align}
Recall that $\sigma$ is defined as
  \begin{align*}
  \sigma(x) = |2x - 1| =
  \begin{cases}
  2x - 1,\quad x \geq 1/2\\
  1 - 2x, \quad x < 1/2,
  \end{cases}
  \end{align*}
and that $0 \leq \sigma(F_k(X_k)) < 1$, $k\in V$.

Consider the summands in~\eqref{eq:ext-1} where $h\in A_{jk}$. We distinguish two cases.
When the ratio $c = \beta_{h\to k} / \beta_{h\to j} > 0$,
we can bound the summand 
  \begin{align*}
  \P\left( \beta_{h\to j}\eps_h > x \right) \geq &\ \E \left[\sigma(F_k(X_k)) \mathbf{1}\left\{ \beta_{h\to j}\eps_h > x \right\}\right] \\
  \geq &\  \E \left[\sigma(F_k(X_k)) \mathbf{1}\left\{ \beta_{h\to k}\eps_h > cx, X_k > cx\right\}\right] \\
  \geq &\ \sigma(F_k(cx)) \P \left( \beta_{h\to k}\eps_h > cx, X_k >  cx\right)\\
  = &\ \sigma(F_k(cx)) \left[\P \left(\beta_{h\to k}\eps_h > cx\right) - \P \left(\beta_{h\to k}\eps_h > cx, X_k \leq cx\right)\right]\\
  = &\ \P \left( \beta_{h\to j}\eps_h > x \right) + o\{\P(X_j > x)\}.
  \end{align*}
The last equality follows from Lemma~\ref{lemma:regvar1} and since $\sigma(F_k(cx)) \to 1$ as $x\to \infty$.
When the ratio $c = \beta_{h\to k} / \beta_{h\to j} < 0$, we obtain
  \begin{align*}
  \P\left( \beta_{h\to j}\eps_h > x \right)
  \geq &\ \E \left[\sigma(F_k(X_k)) \mathbf{1}\left\{ \beta_{h\to j}\eps_h > x \right\}\right]\\
  \geq &\  \E \left[\sigma(F_k(X_k)) \mathbf{1}\left\{ \beta_{h\to k}\eps_h  < cx, X_k < cx\right\}\right] \\
  \geq &\ \sigma(F_k(cx)) \P \left( \beta_{h\to k}\eps_h < cx, X_k <  cx\right)\\
  = &\ \sigma(F_k(cx)) \left[\P \left(\beta_{h\to k}\eps_h < cx\right) - \P \left(\beta_{h\to k}\eps_h < cx, X_k \geq cx\right)\right]\\
  = &\ \P \left( \beta_{h\to j}\eps_h > x \right) + o\{\P(X_j > x)\},
  \end{align*}
where the third inequality holds because, as $x \to \infty$,
$F_k(X_k) < F_k(cx)$ implies $\sigma(F_k(X_k)) > \sigma(F_k(cx))$.

On the other hand, for all the summands in~\eqref{eq:ext-1} where $h \in A_{jk^{*}}$ we have that $X_k$ and $\eps_h$ are independent. Therefore,
  \begin{align*}
    \E \left[\sigma(F_k(X_k)) \mathbf{1}\left\{ \beta_{h\to j}\eps_h > x \right\}\right]
    = &\ \E \left[|2F_k(X_k) - 1|\right] \P \left(\beta_{h\to j}\eps_h > x\right)\\
    = &\ \frac12  \P\left( \beta_{h\to j}\eps_h > x \right).
  \end{align*}
Consequently,
  \begin{align*}
    \Psi_{jk}^{+} = &\lim_{x\to \infty}\frac{1}{2}\E\left[\sigma(F_k(X_k)) \mid X_j > x\right]\\
    = &\  \lim_{x\to \infty} \frac{1}{2} \sum_{h\in A_{jk}} \frac{\P\left(\beta_{h\to j}\eps_h > x \right)}{\P(X_j > x)} + \lim_{x\to \infty} \frac{1}{4} \sum_{h\in A_{jk^{*}}} \frac{\P\left(\beta_{h\to j}\eps_h > x \right)}{\P(X_j > x)}\\
    = &\ \frac{1}{4} + \frac{1}{4} \sum_{h\in A_{jk}}\lim_{x\to\infty}\frac{\P(\beta_{h\to j}\eps_h > x)}{\P(X_j > x)}
    =  \frac{1}{4} + \frac{1}{4} \frac{\sum_{h\in A_{jk}} |\beta_{h\to j}|^{\alpha} c_{hj}^{+}\ell(x)x^{-\alpha}}{\sum_{h\in\An(j, G)} |\beta_{h\to j}|^{\alpha}c_{hj}^{+}\ell(x)x^{-\alpha}}\\
    = &\ \frac{1}{4} + \frac{1}{4} \frac{\sum_{h\in A_{jk}} c_{hj}^{+} |\beta_{h\to j}|^{\alpha}}{\sum_{h\in\An(j, G)} c_{hj}^{+} |\beta_{h\to j}|^{\alpha}}.
  \end{align*}

Similarly, the second term can be shown to be
  \begin{align*}
  \Psi_{jk}^{-} =  \frac{1}{4} + \frac{1}{4} \frac{\sum_{h\in A_{jk}} c_{hj}^{-} |\beta_{h\to j}|^{\alpha}}{\sum_{h\in\An(j, G)} c_{hj}^{-} |\beta_{h\to j}|^{\alpha}}.
  \end{align*}

Putting everything together yields the desired form of $\Psi_{jk}=  \Psi_{jk}^{+} + \Psi_{jk}^{-}$.
\end{proof}

\section{Example of the EASE algorithm}\label{app:minimax_search}
Consider the DAG $G$ in Figure~\ref{fig:tree_search_dag}, with vertex set $V = \{1, 2, 3, 4\}$. The set of causal orders of $G$ is $\Pi_{G} = \{ (2, 1, 4, 3), (2, 1, 3, 4), (3, 1, 2, 4)\}$.
In Figure~\ref{fig:tree_search}, we display the state-space tree of the EASE algorithm, i.e., the set of states that the algorithm can visit to find a causal order $\pi$ of $G$.
Each state represents the status of the vector $\pi^{-1}$ during the algorithm evaluation.
A state is red if all the paths below it lead to wrong causal order.
The green states represent the causal orders of $G$.
\begin{figure}[!ht]
\centering
\includegraphics{tree_search_dag.tex}
\caption{DAG $G$.}
\label{fig:tree_search_dag}
\end{figure}
\begin{figure}[h!]
\centering
\includegraphics{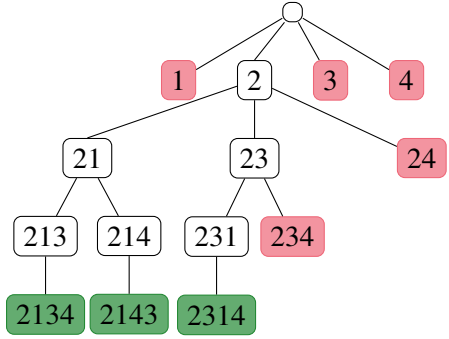}
\caption{Extremal ancestral search (EASE) for the DAG shown in Figure~\ref{fig:tree_search_dag}.}
\label{fig:tree_search}
\end{figure}

\newpage
\section{Experimental settings for the simulation study}\label{app:simulation_experiments}

The parameters of the simulation are the following.
\begin{itemize}
\item[--] Distribution: Student's $t$, with degrees of freedom $\alpha \in \{1.5, 2.5, 3.5\}$.
\item[--] Number of observations: $n \in \{500, 1000, 10000\}$.
\item[--] Number of variables: $p \in \{4, 7, 10, 15, 20, 30, 50\}$.
\end{itemize}

The settings that we consider are,
\begin{itemize}
\item[--] Linear SCM,
\item[--] Linear SCM with hidden confounders,
\item[--] Nonlinear SCM,
\item[--] Linear SCM and uniform transformation of each variable.
\end{itemize}

For each combination of $n$, $p$, and $\alpha$ and each setting, we generate $n_{exp} = 50$ random SCMs. Each SCM is built as follows.
\begin{enumerate}
\item Generate a random DAG.
  \begin{enumerate}
  \item Take a random permutation $\pi$ of the nodes $V = \oneto{p}$ that defines the causal order.
  \item For each $i\in V$ such that $\pi(i) > 1$, sample the number of parents
  $$n_{\pa} \sim \text{Bin}(\pi(i) - 1, q),$$ from a binomial distribution. We set $q = \min\{5/(p-1), 1/2\}$ so that,  on average, there are 2.5 edges per node, when $p > 10$.
  \item Sample without replacement $n_{\pa}$ from $\{j\in V: \pi(j) < \pi(i)\}$ and name the resulting set $\pa(i)$.
  \end{enumerate}
\item Sample uniformly from $\{-0.9, -0.1\}\cup\{0.1, 0.9\}$ the coefficients $\beta_{ij}$, where $i\in V$ and $j\in \pa(i, G)$.
\item In the case of hidden confounders,
  \begin{enumerate}
  \item Sample the number of confounding variables,
  $$n_{\text{conf}} \sim \text{Bin}\left(\frac{p(p-1)}{2}, q\right),$$
  from a binomial distribution. We set $q = 2/(3p - 3)$ so that, on average, there is one hidden confounder for every three nodes.
  \item Sample without replacement $n_{\text{conf}}$ unordered pairs from $\left\{\{i, j\}: i, j \in V\right\}$ and name the resulting set $C$.
  \item  Update the parents of each node $i$ as $\pa(i) := \pa(i) \cup C_i$, where $C_i\subseteq C$ is the set of hidden confounders affecting node $i\in V$. Similarly, for each hidden confounder $c \in C$, set $\pa(c) = \varnothing$.
  \item Sample uniformly from $\{-0.9, -0.1\}\cup\{0.1, 0.9\}$ the coefficients $\beta_{ic}$, where $i\in V$ and $c\in C_i$.
  \end{enumerate}
\item Let $\widetilde V = V \cup C$\footnote{If there are no hidden confounders, then $C = \varnothing$.}. For all $i\in \widetilde V$, sample $n$ i.i.d. copies of noise $\eps_i \sim \text{Student's } t\ \text{with df} = \alpha$.
\item For each node $i\in \widetilde V$, generate
      $$X_i := \sum_{j\in\pa(i)} \beta_{ij}\ f(X_j) + \eps_i,\quad\text{where}$$
  \begin{enumerate}
  \item in case of linear SCM, $f(X_j) = X_j$,

  \item in case of nonlinear SCM, $f(X_j) = X_j\ \mathbf{1}\{\widehat F_j(X_j) > 0.95\}$, where $\widehat F_j$ is the empirical cdf of $X_j$.

  \end{enumerate}

\item In case of uniform margins, transform each variable by $X_i := \widehat F_i(X_i)$, where $\widehat F_i$ is the empirical cdf of $X_i$, $i\in \widetilde V$.
\end{enumerate}
	\newpage
	\section{Additional Figures and Tables}\label{app:additional_figs}
	~
	
\begin{figure}[h!]
\centering
\includegraphics[scale=0.45]{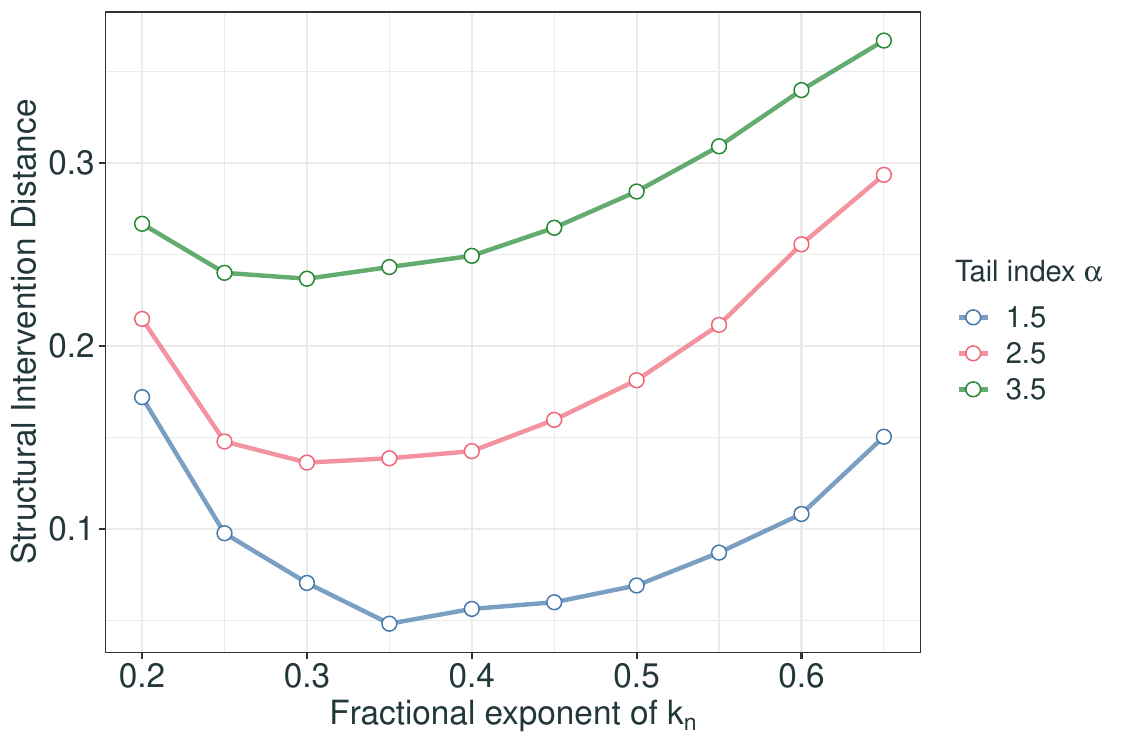}
\caption{The figure refers to Section~\ref{subsec:sim_results}. It shows the structural intervention distance (SID) of the EASE algorithm for different fractional exponents $\nu \in [0.2, 0.7]$ of $k_n = \lfloor n^\nu \rfloor$ and different tail indices $\alpha \in\{1.5, 2.5, 3.5\}$. Each point represents the SID measure averaged over 10 random samples for different sample sizes $n \in \{500, 1000, 10000\}$ and dimensions $p \in \{4, 7, 10, 15, 20, 30, 50\}$ in a linear SCM. 
In the experiments of Section~\ref{sec:simulation_study}, we set $k_n = \lfloor n^{0.4}\rfloor$.}
\label{fig:k-robustness}
\end{figure}

\begin{figure}[h!]
\centering
\includegraphics[scale=.45]{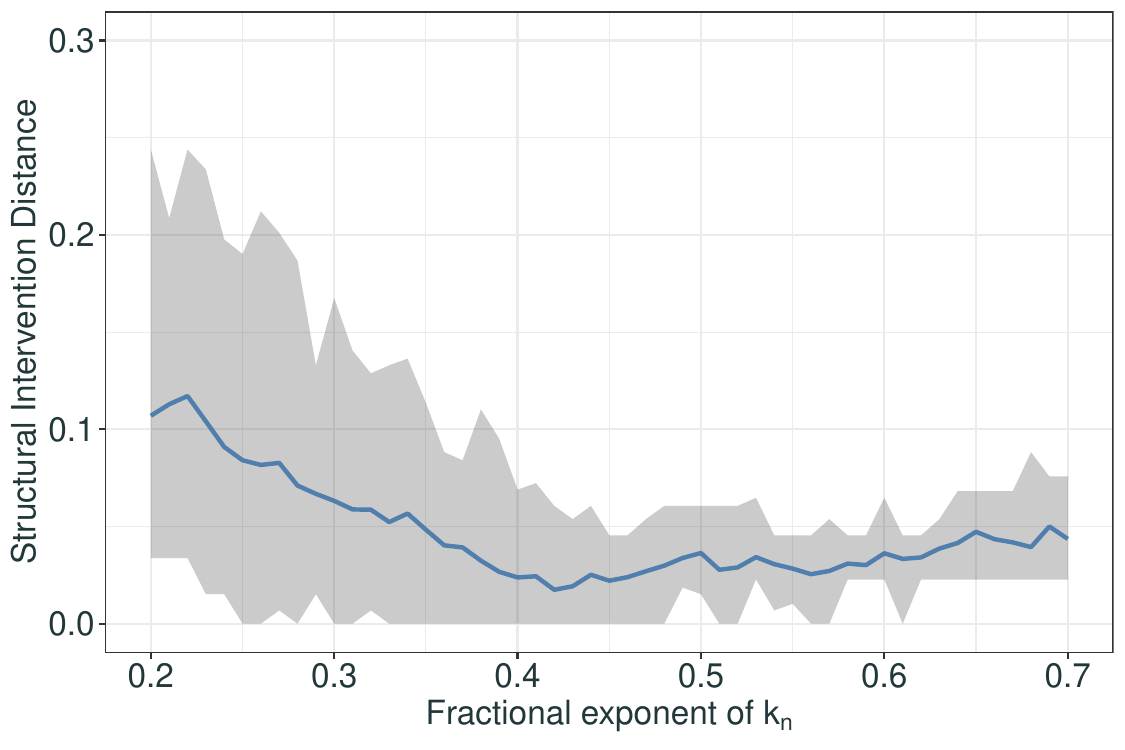}
\caption{The figure refers to Section~\ref{sec:river}. It shows the robustness of the structural intervention distance (SID) for varying fractional exponents $\nu \in [0.2, 0.7]$ of $k_n = \lfloor n ^ \nu\rfloor$.
Each point represents an average over 50 SID evaluations for the EASE algorithm, after bootstrapping the original dataset. The shaded interval corresponds to the 90\% bootstrap confidence interval.}
\label{fig:river_robk}
\end{figure}

\begin{figure}[h!]
\centering
\includegraphics[scale=0.6]{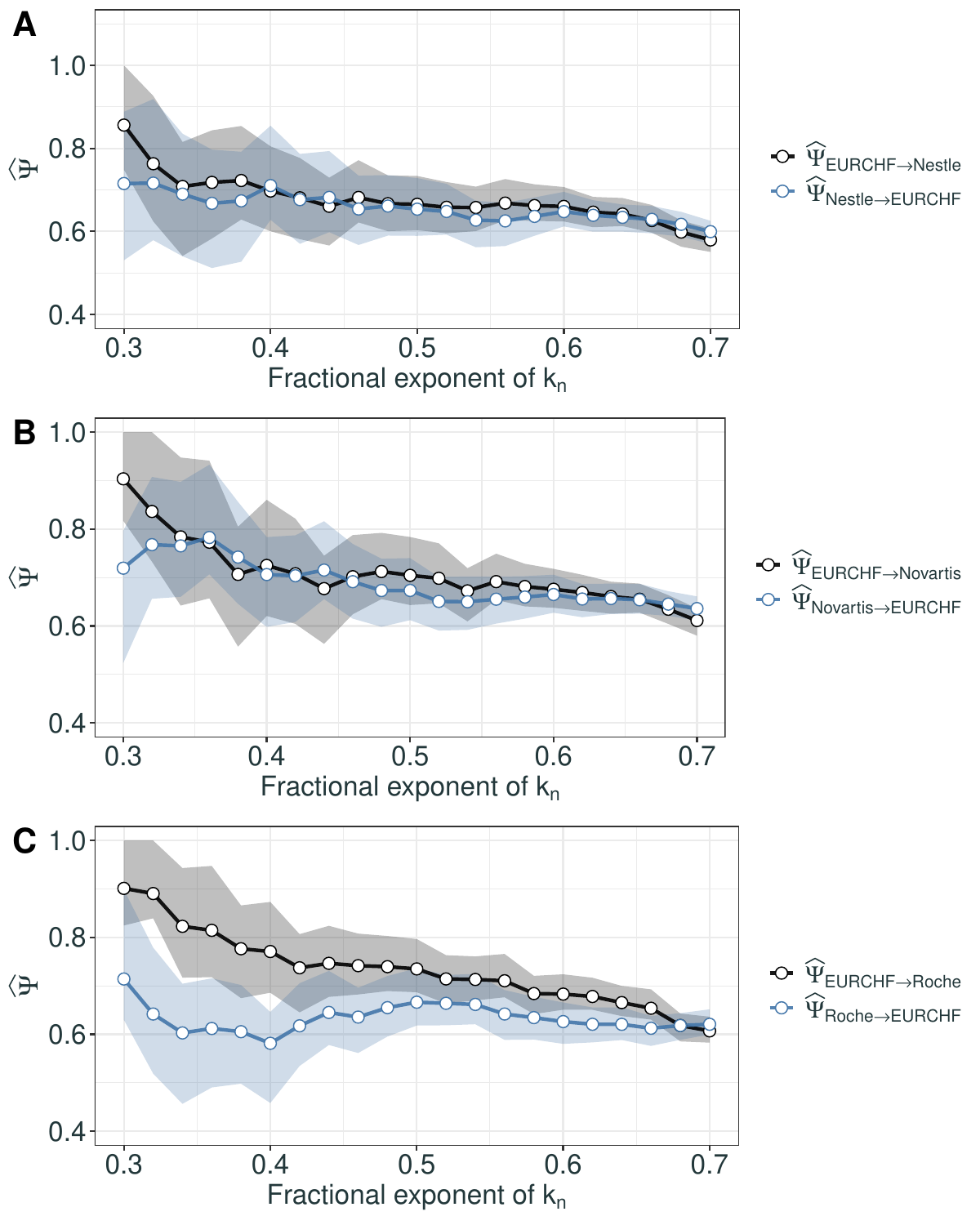}
\caption{The figure refers to Section~\ref{sec:financial}. It shows the variability of the estimated coefficients $\widehat\Psi$ for different fractional exponents $\nu \in [0.2, 0.7]$ of $k_n = \lfloor n ^ \nu \rfloor$. Each point represents the estimates of $\widehat\Psi$
based on the full dataset.
The shaded intervals correspond to the 90\% bootstrap confidence intervals
over 1000 repetitions.}
\label{fig:financial_robk}
\end{figure}

	\newpage
	~
	
\begin{figure}[H]
\centering
\includegraphics[scale=0.55]{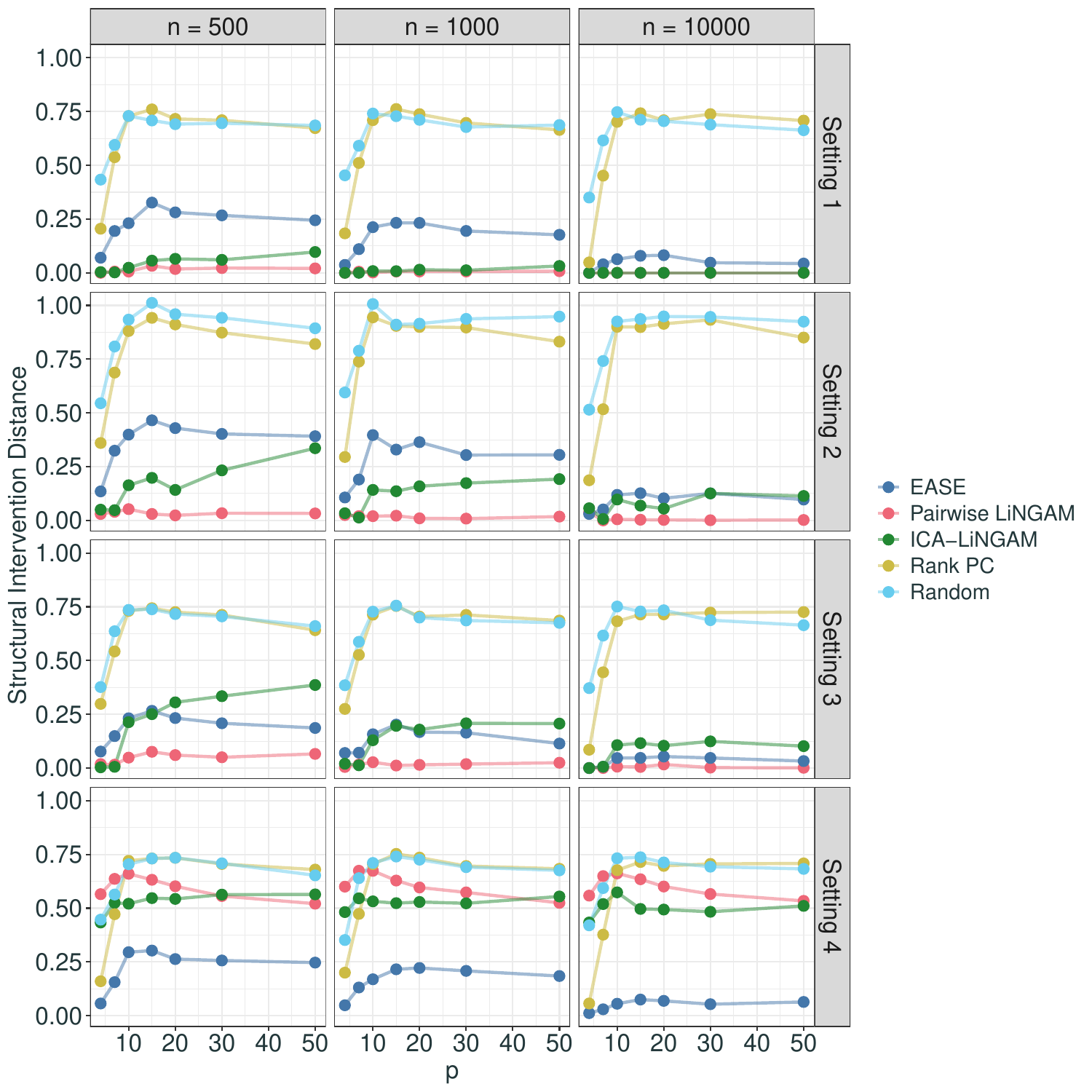}
\caption{The figure refers to Section~\ref{subsec:sim_results}. It shows the SID averaged over 50 simulations, for each method, setting, sample size $n$ and dimension $p$, when the tail index is \boldmath$\alpha = 2.5$.
Each row of the figure corresponds to one setting. In order, the settings are: (1)~Linear SCM; (2)~Linear SCM with hidden confounders; (3)~Nonlinear SCM; (4)~Linear SCM where each variable is transformed to a uniform margin.}
\label{fig:sid_plot_alpha_25}
\end{figure}

\begin{figure}[H]
\centering
\includegraphics[scale=0.55]{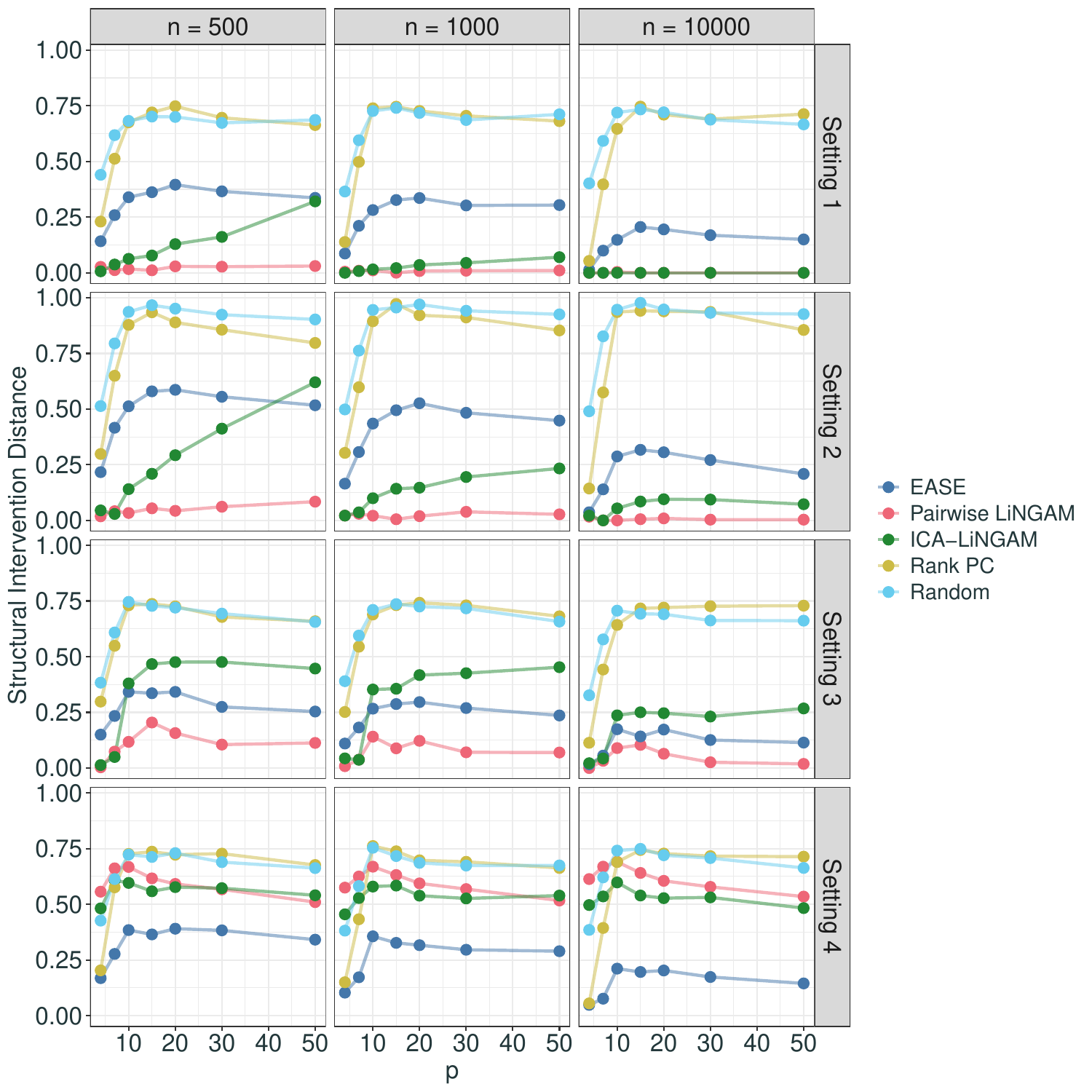}
\caption{The figure refers to Section~\ref{subsec:sim_results}. It shows the SID averaged over 50 simulations, for each method, setting, sample size $n$ and dimension $p$, when the tail index is \boldmath$\alpha = 3.5$.
Each row of the figure corresponds to one setting. In order, the settings are: (1)~Linear SCM; (2)~Linear SCM with hidden confounders; (3)~Nonlinear SCM; (4)~Linear SCM where each variable is transformed to a uniform margin.}
\label{fig:sid_plot_alpha_35}
\end{figure}

\begin{figure}[H]
\centering
\includegraphics[scale=0.6]{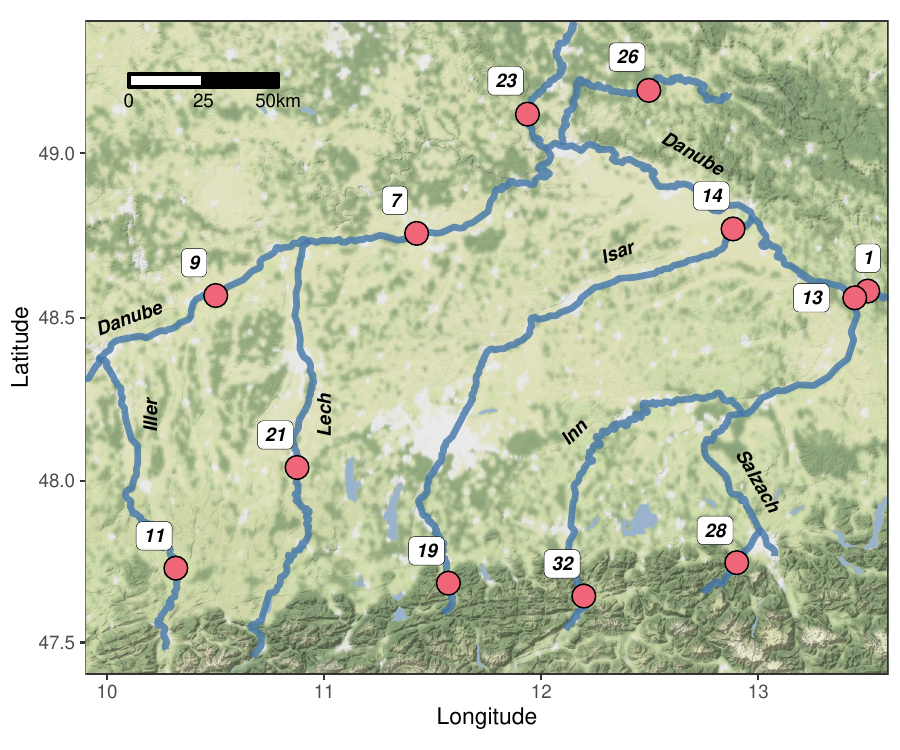}
\caption{The figure refers to Section~\ref{sec:river}. It represents the map of the upper Danube basin. The plot is created with the \href{https://cran.r-project.org/web/packages/ggmap/index.html}{ggmap} \texttt{R} package developed by~\citet{wickhamggmap}. The background is taken from \href{maps.stamen.com}{maps.stamen.com}.}
\label{fig:river_map}
\end{figure}

\begin{figure}[H]
\ifjournal{\hspace*{-.8cm}}\fi
\centering
\includegraphics[scale=0.6]{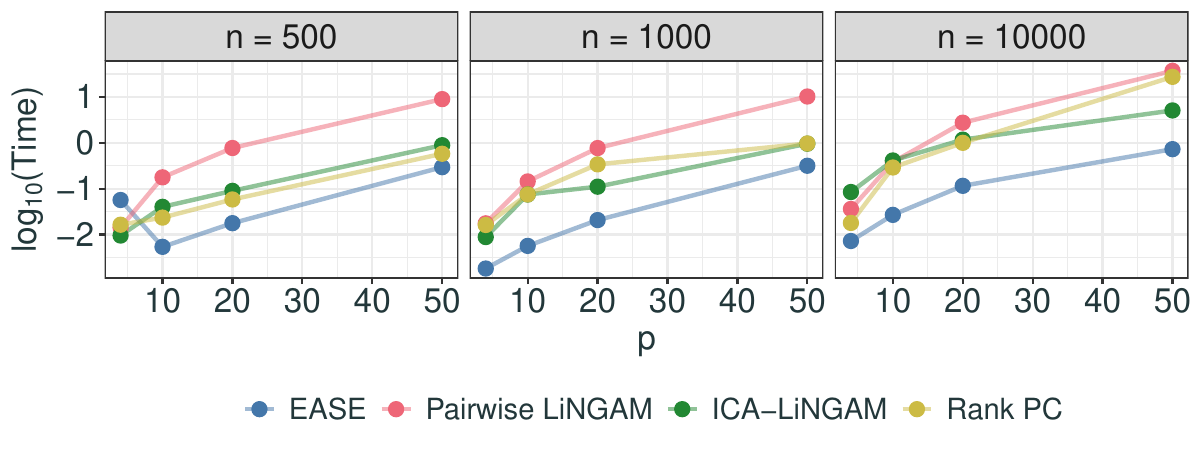}
\caption{The figure refers to Section~\ref{subsec:sim_results}. It shows the base-10 logarithm
of the computational time (in seconds) averaged over 10 simulations, for each method, sample size $n$ and dimension $p$.}
\label{fig:time}
\end{figure}

\begin{table}[H]
\caption{\label{tab:lingam_experiment}
The table refers to Section~\ref{subsec:sim_results}.
It displays the average SID and the standard error
(SE) over 50 simulations, for ICA-LiNGAM and Pairwise LiNGAM.
For each dimension $p$ we simulate $n = 10000$ observations from a non-linear SCM with Student's-$t$ noise with $\alpha = 3.5$ degrees of freedom
(see Section~\ref{sec:simulation_study}).
We run each method both on the full dataset (`Full dataset') and on the
partial dataset, where we keep only the observations in the tails of the
distribution (`Keep tails').
As dimension $p$ increases, the number of extreme observations
in all their coordinates decreases exponentially.
Therefore, for a given dimension $p$,
we say that an observation $x\in\mathbb{R}^p$ lies
in the tails of the distribution
if at least $\lfloor\sqrt{p}\rfloor$ of its coordinates are
below (or above) the 10\% (or 90\%) quantile.
}
\centering
\begin{tabular}{llcccc}
\toprule
\multicolumn{2}{c}{ } & \multicolumn{2}{c}{\emph{Full dataset}} & \multicolumn{2}{c}{\emph{Keep tails}} \\
\cmidrule(l{3pt}r{3pt}){3-4} \cmidrule(l{3pt}r{3pt}){5-6}
 &  & \emph{SID} & \emph{SE} & \emph{SID} & \emph{SE}\\
\midrule
 & $p = 10$ & 0.236 & 0.029 & 0.110 & 0.021\\
 & $p = 20$ & 0.246 & 0.022 & 0.150 & 0.017\\
 & $p = 30$ & 0.231 & 0.017 & 0.148 & 0.014\\
\multirow{-4}{*}{\raggedright\arraybackslash ICA-LiNGAM} & $p = 50$ & 0.268 & 0.014 & 0.216 & 0.012\\
\\
 & $p = 10$ & 0.090 & 0.024 & 0.003 & 0.001\\
 & $p = 20$ & 0.064 & 0.013 & 0.003 & 0.001\\
 & $p = 30$ & 0.026 & 0.007 & 0.005 & 0.003\\
\multirow{-4}{*}{\raggedright\arraybackslash Pairwise LiNGAM} & $p = 50$ & 0.019 & 0.004 & 0.006 & 0.002\\
\bottomrule
\end{tabular}
\end{table}

\begin{figure}[H]
\ifjournal{\hspace*{-.8cm}}\fi
\centering
\includegraphics[scale=0.6]{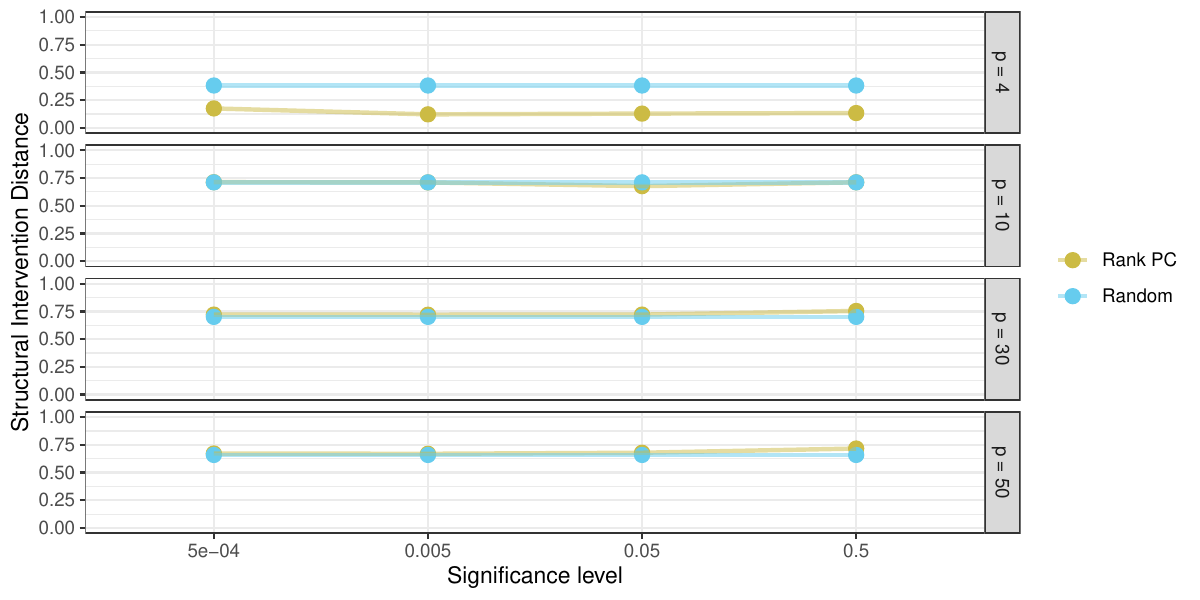}
\caption{The figure refers to Section~\ref{subsec:sim_results}. It shows the average SID of the Rank PC algorithm over 50 simulations for different significance levels of the independent test.  For each dimension $p$ we simulate $n = 1000$ observations from a linear SCM with Student's-$t$ noise with $\alpha = 3.5$ degrees of freedom (see Section~\ref{sec:simulation_study}). For comparison, we also display the SID associated with random guessing.}
\label{fig:rankpc}
\end{figure}

\section{Financial application}\label{app:financial_dynamics}
Starting from the financial application of Section~\ref{sec:financial}, we study the dynamic evolution of the $\widehat\Psi$ coefficient across time.
By looking at the time series of the EURCHF in Figure~\ref{fig:eurchf_plot}, we notice two periods of extremely high positive and negative returns, namely, August 2011 and January 2015. The second event was due to an unexpected decision of the Swiss National Bank (SNB) to remove the peg of 1.20 Swiss francs per Euro.

\begin{figure}[H]
\centering
\includegraphics[scale=0.4]{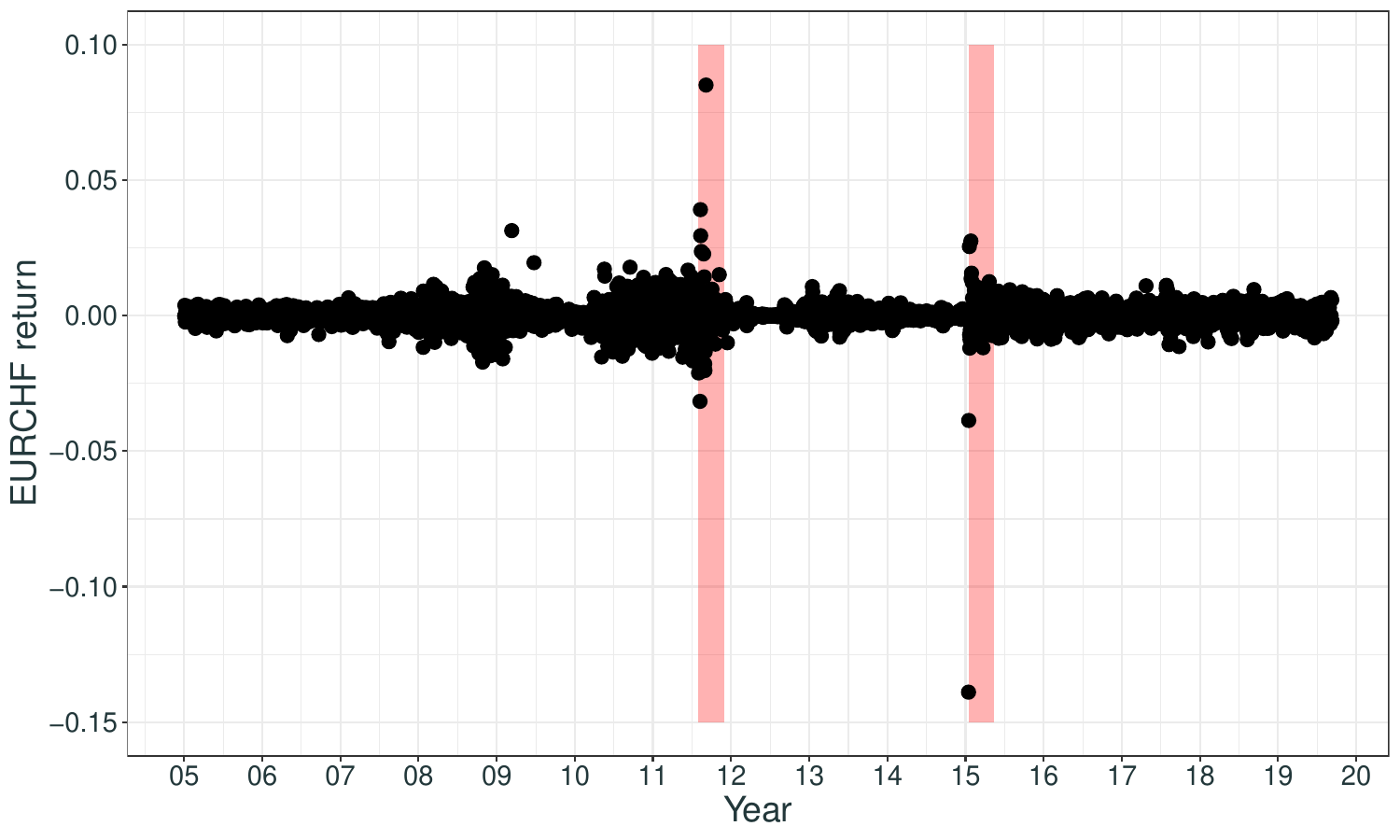}
\caption{Daily returns of the EURCHF exchange rate. The red sections represent the two major turmoil events. They occurred, respectively, in the month of August 2011 and in January 2015. The latter event was caused by the unexpected decision of the Swiss National Bank (SNB) to remove the peg of 1.20 Swiss francs per Euro.}
\label{fig:eurchf_plot}
\end{figure}

The idea is to estimate the $\widehat\Psi$ coefficients between the EURCHF and the three stocks on a rolling window of 1500 days. In this case, we use a threshold $k = 10$ which corresponds approximately to a fractional exponent $\nu = 0.3$, where we define $k = \lfloor n ^ \nu\rfloor$ and $n$ is the number of observations in the sample. In Figure~\ref{fig:rolling_eurchf}, we notice that during turmoil periods, highlighted in red, the $\widehat\Psi$ coefficient is higher in the direction that goes from the EURCHF to the stocks. This is an example where the causal structure becomes clearer during extreme events. It is also interesting to observe that the $\widehat\Psi$ coefficient is quite stable during calmer market periods. For example, the stable black lines between 2012 and 2015 correspond to the currency peg maintained by the SNB. As a last note, the drop in the black lines in mid 2017 is due to the fact that the extreme event of 2011 `exits' the rolling window.

\begin{figure}[H]
\centering
\includegraphics[scale=0.6]{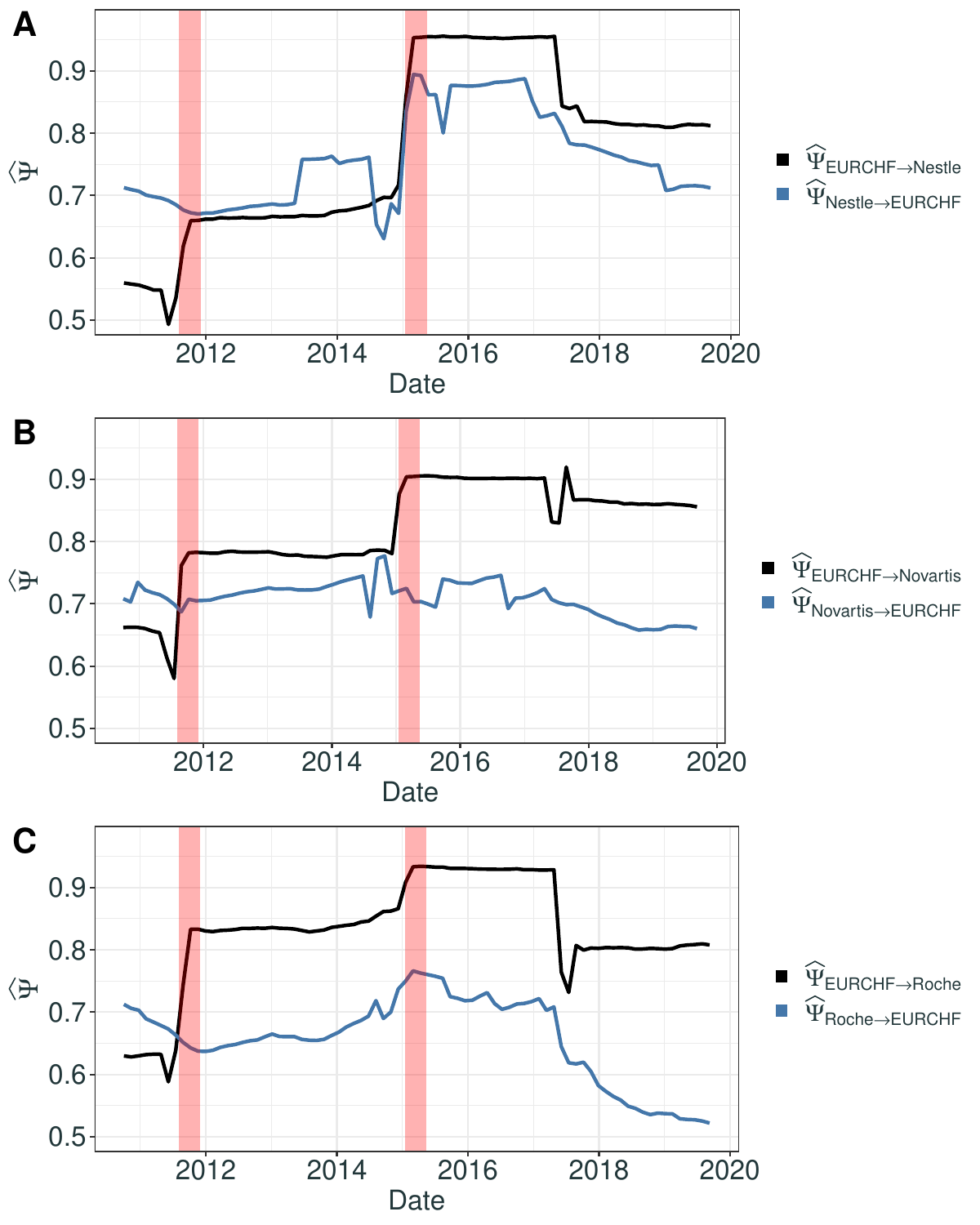}
\caption{Estimated coefficients $\widehat\Psi$ on a rolling window of 1500 days. The threshold used to estimate the coefficient is $k = 10$.}
\label{fig:rolling_eurchf}
\end{figure}

}\fi

\end{appendix}

\clearpage
\ifjournal{
\bibliographystyle{imsart/imsart-nameyear}}
\else{
\bibliographystyle{plainnat}
}\fi
\bibliography{ref.bib}

\end{document}